\documentclass[11pt,notitlepage]{article}

\usepackage[affil-it]{authblk}
\usepackage[colorlinks,allcolors=blue]{hyperref}
\usepackage[none]{hyphenat}

\usepackage{graphicx}%
\usepackage{multirow}%
\usepackage{amsmath,amssymb,amsfonts}%
\usepackage{amsthm}%
\usepackage{mathrsfs}%
\usepackage[title]{appendix}%
\usepackage{xcolor}%
\usepackage{textcomp}%
\usepackage{manyfoot}%
\usepackage{booktabs}%
\usepackage{algorithm}%
\usepackage{algorithmicx}%
\usepackage{algpseudocode}%
\usepackage{listings}%
\usepackage{extarrows}
\usepackage{tikz-cd}

\usepackage[left=2.5cm,top=2.5cm,right=2.5cm,bottom=2.5cm]{geometry}

\newtheorem{theorem}{Theorem}
\newtheorem{proposition}[theorem]{Proposition}

\newtheorem{definition}{Definition}%

\newcommand{\Si}{{\Sigma}}
\newcommand{\dm}{{\partial M}}
\newcommand{\hSi}{{\hat{\Sigma}}}
\newcommand{\hsi}{{\hat{\sigma}}}
\newcommand{\hdm}{{\hat{\partial M}}}
\newcommand{\Hi}{{\mathcal{H}}}
\newcommand{\dcg}{{\mathbf{Bord}_{d}^{q,G}}}

\newcommand{\coh}[3][2]{H^{#1}(#2,#3)}
\newcommand{\homo}[3][2]{{H_{#1}(#2,#3)}}

\renewcommand{\hat}{\widehat}

\newcommand{\bigslant}[2]{{\raisebox{.2em}{$#1$}\left/\raisebox{-.2em}{$#2$}\right.}}

\def \Hom{\mathrm{Hom}}
\def \R{\mathbb{R}}
\def \C{\mathbb{C}}

\def \Z{\mathbb{Z}}

\def \tn{{\mathrm{n}}}
\def \tg{{\mathrm{g}}}

\def \Bord{\mathbf{Bord}}

\newcommand{\sk}{{\mathrm{Sk}}}

\title{On finite group global and gauged $q$-form symmetries in TQFT}

\author[1]{Manuel Furlan}
\author[2]{Pavel Putrov}
\affil[1]{Dipartimento di Fisica, Universit\'a di Trieste, Strada Costiera 11, 34151 Trieste, Italy}
\affil[2]{ICTP, Strada Costiera 11, 34151 Trieste, Italy}
\date{}

\begin{document}

\maketitle

\begin{abstract} 
We describe a method to implement finite group global and gauged $q$-form symmetries into the axiomatic structure of $d$-dimensional Topological Quantum Field Theory (TQFT) in terms of bordisms decorated by cohomology classes. Namely, on a manifold with a boundary, the gauge field is considered as a class in an appropriate relative cohomology group. It is defined in a way that allows self-consistent cutting and gluing of the manifolds and involves a choice of a $(d-q-2)$-skeleton in the boundary. The method, in a sense, generalizes to arbitrary $d$ and $q$ a method that has been considered in the literature in the case of $d=3,\;q=0,1$.

\end{abstract}

\tableofcontents

\section{Introduction and summary}\label{sec1}
Since the seminal article \cite{Gaiotto}, generalized global symmetries have been a very active area of research in the theoretical and mathematical physics community. The basic idea is to understand symmetries in quantum field theories in terms of extended topological operators also referred to as ``defects''. In this approach, the usual charge operator of an ordinary, i.e. 0-form in modern terminology, symmetry is a codimension 1 topological defect supported on a spatial slice of the spacetime manifold.

More generally, for a group $G$, one can consider an invertible $q$-form symmetry in a $d$-dimensional spacetime as being described by $(d-q-1)$-dimensional topological defects labeled by the elements of the group. The fusion of such defects obeys the group law of $G$. In this paper we will restrict ourselves to such invertible symmetries for a finite abelian group $G$ (which must be abelian for $q>0$). One can consider in general a ``network'' of such topological defects. Mathematically such a network can be understood as a singular cycle in the spacetime manifold $M$ with coefficients in $G$.  For a closed spacetime manifold $M$, and the case of a non-anomalous symmetry, the fact that the defects are ``topological'' means that the partition function (without insertion of any other operators) of the quantum field theory depends only on the homology class of the cycle in $H_{d-q-1}(M;G)$. One can also consider its  Poincar\'e dual class $B\in H^{q+1}(M;G)$, which can be understood as a background gauge field for the discrete $q$-form symmetry $G$.

When $M$ has a non-trivial boundary $\partial M$ it becomes more subtle to describe such background gauge fields in terms of cohomology classes (or, equivalently, the defect networks in terms of homology classes). Of course, the description in principle should depend on what type of boundary condition on the gauge fields/defects one wants to consider. It is natural to ask what boundary conditions allow gluing different manifolds along boundary components, such that an arbitrary field/defect configuration can be realized in the resulting glued manifold. 

It is easy to see that the two most naive choices $B\in H^{q+1}(M;G)\cong H_{d-q-1}(M,\partial M;G)$ or $B\in H^{q+1}(M,\partial M;G)\cong H_{d-q-1}(M;G)$ cannot work for this purpose, as illustrated schematically in the Figure \ref{fig:intro1}.

    \begin{figure}
        \centering
        \includegraphics[width=0.5\textwidth]{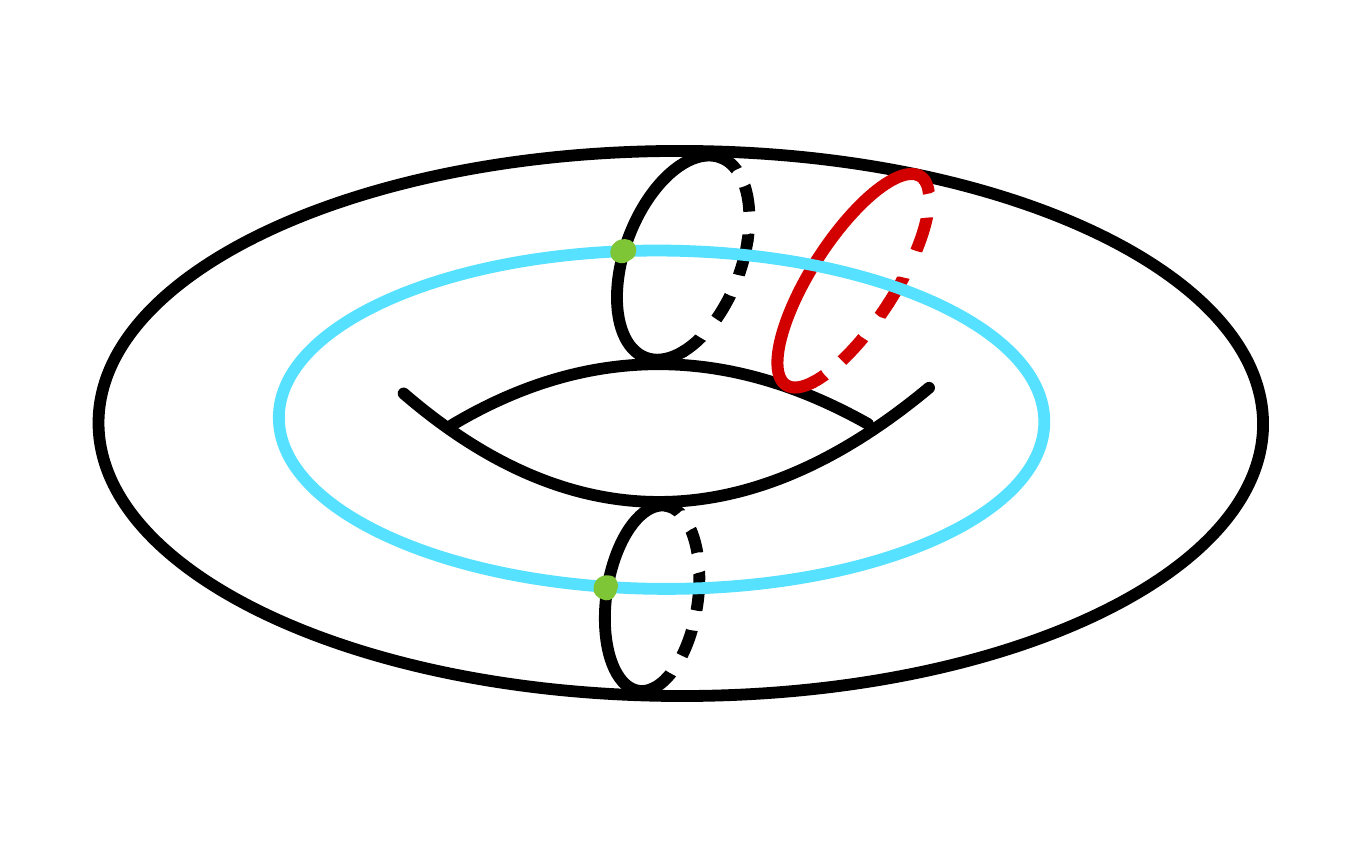}
        \caption{The case of $d=2$ and $q=0$. The 2-manifold $M=S^1\times S^1$ is obtained by gluing two cylinders $M_\pm\cong S^1 \times I$ (where $I:=[0,1]$ is an interval) along $S^1\sqcup S^1$. The red defect is an example of a defect that cannot be obtained from the defects in $M_\pm$ representing $H_1(M_\pm,\partial M_\pm;G)\cong H^1(M_\pm;G)$. The blue defect is an example of a defect that cannot be obtained from the defects in $M_\pm$ representing $H_1(M_\pm;G)\cong H^1(M_\pm,\partial M_\pm;G)$. }
        \label{fig:intro1}
    \end{figure}   

The main goal of this paper is to show systematically that taking instead $B\in H^{q+1}(M,\hat{\partial{M}};G)$, where $\hat{\partial M }:=\partial M\setminus \sk$ is the complement of a $(d-q-2)$-skeleton $\sk$ in the boundary $\partial M$, does the job. In particular, this allows us to define a category of bordisms decorated by classes in such relative cohomology groups. A TQFT with a global $q$-form symmetry $G$ then can be formally defined as a symmetric monoidal functor on it. This approach, in a sense, generalizes the one considered in \cite{blanchet2016non,Costantino:2021yfd,Benini:2022hzx}. There, the focus was on the case of $d=3$ and $q=0$ and $q=1$. In those works one chooses a single point in each connected component of the 2-dimensional boundary. Then, for $q=1$, one considers cohomology groups relative to complement of that collection of points in the boundary. The collection of points thus can be understood as the 0-skeleton of our general approach, where for a boundary component surface of genus $g$ one considers a standard cell decomposition consisting of a single 2-cell, a single 0-cell, and $2g$ 1-cells. For $q=0$ one considers cohomology groups relative to that collection of points. The 1-skeleton of the general approach then can be understood as the deformation retract of the complement of that point. 

A known common description of TQFTs with a global or gauged symmetry is via the homotopy quantum field theory \cite{turaev2010homotopy,Dijkgraaf:1989pz,Freed:1991bn}. In this setting, a TQFT with a global $q$-form symmetry $G$ can be defined as a functor on the category of bordisms equipped with maps to the Eilenberg-MacLane space $K(G,q+1)$, considered up to homotopy relative to the boundary. The composition of the morphisms is then defined in a straightforward way. However, one has to keep track of the restriction of the map to the boundary. Using relative cohomology class $B\in H^{q+1}(M,\hat{\partial{M}};G)$ instead of maps to $K(g,q+1)$ provides a much more expicit algebraic description. The relation between the two approaches can be schematically understood as follows. The choice of the $(d-q-2)$-skeleton corresponds to considering the maps to $K(G,q+1)$ that have the restriction to the boundary of a particular form. Namely, the maps are constant outside of a small neighborhood of the skeleton, which can be locally taken of the form $\R^{d-q-2}\times D^{q+1}$. In the neighborhood, the map can be considered locally to a product map, constant along $\R^{d-q-2}$ and such that the map on $D^{q+1}$ is a fixed representative of a homotopy class in $\pi_{q+1}(K(G,q+1))=G$, where the boundary of $D^{q+1}$ is identified with the basepoint of $S^{q+1}$.  

One can realize the procedure of gauging as an operation on a functor on the category of bordisms decorated by cohomology classes. The result is a new TQFT functor. The gauging consists of summing over the relative cohomology classes, with an appropriate overall factor. The factor gets rid of ``overcounting'' of field configurations when one glues manifolds along the common boundary. This makes the value of the new TQFT, i.e. the partition function of the gauge theory, functorial with respect to the composition of the bordisms. Specifically, the formula for the partition function of the gauged theory on a $d$-manifold $M$ with boundary $\partial M$, in terms of the partition function of the theory with global symmetry reads:
\begin{equation}
   Z_{\text{g},b}(M) = c(M) \sum_{\substack{B \in \coh[q+1]{M}{\hdm}\\ B|_{\partial M}=b}} Z(M,B)
   \label{gauging-intro}
\end{equation}
where
\begin{equation}
    c(M)=\prod_{i=0}^q \left( |\coh[i]{M}{\hdm} | |H^{i}({\partial M}) |^{\frac{1}{2}} \right )^{(-1)^{q-i+1}} 
\end{equation}
and all the cohomology groups have coefficients in $G$. When $\partial M=\emptyset$ the formula reduces to a well-known expression for the partition function of the gauged theory on a closed manifold. The formula above can be also refined by introducing a non-trivial background of the dual $(d-q-2)$-form dual symmetry $G^*=\Hom(G,\R/\Z)$.

The approach to the TQFTs with global and gauged symmetries considered in this paper has some similarities with the approach considered in \cite{Carqueville_2019,carqueville2023orbifolds,Carqueville:2018sld} where more general ``defect TQFT'' and their gauging is considered. In our setting, due to the restriction to invertible defects of a fixed dimension, a more explicit algebraic description can be provided.

The rest of the paper is organized as follows. In Section \ref{sec2} we provide the definition of the bordism category of manifolds decorated by cohomology classes and a TQFT with a global $q$-form symmetry $G$ as a functor $\mathcal{Z}$ on this category. In Section \ref{sec3} we consider the gauging procedure as an operation on taking the TQFT functor $\mathcal{Z}$ as an input and having as an output the gauged TQFT functor $\mathcal{Z}_\tg$. In Section \ref{sec4} we describe how the gauged TQFT $\mathcal{Z}_\tg$ can be understood as a TQFT with a dual $G^*$ global symmetry (in the sense of Section \ref{sec2}). Finally, in Appendix \ref{secA1} we collect various details about the homology and cohomology groups that are used in the main text.

\section{TQFT with a global \texorpdfstring{$q$}{q}-form symmetry}\label{sec2}

As mentioned in the introduction we consider background gauge fields for a $q$-form symmetry $G$, where $G$ is a finite abelian group, as classes of the relative cohomology $\coh[q+1]{M}{\hdm}$ where $\hdm = \dm \setminus \sk$ and $\sk$ is a $(d-q-2)$-skeleton in the (not necessarily connected) boundary $\dm$. Unless explicitly indicated otherwise, the coefficients of the cohomology are assumed to be in $G$ throughout the paper. Also, $\hat{\Sigma}$ by default means complement of the $(d-q-2)$-skeleton in a $(d-1)$-manifold $\Sigma$. To define the skeleton we will assume a certain cellular decomposition of the boundary manifold, for example, a triangulation. In Section \ref{sec4}, for technical reasons, we will restrict ourselves to skeletons  satisfying a certain property (a sufficient condition for this property is that the skeleton is defined by a triangulation or a dual polyhedral decomposition). We will argue that the ``physical information'' is actually independent of the choice of skeleton (and, therefore, cell decomposition).

In this section we will define a TQFT in the style of \cite{aty}, but with global $q$-form symmetry $G$ and as a functor from the bordism category of manifolds decorated by cohomology classes to the category of vector spaces. However, we will have first to check that such refined bordism category is well-defined.

    \subsection{Bordism category of manifolds decorated by cohomology classes}\label{sub2.1}

    In this section, we provide a definition of the \textit{bordism category of manifolds decorated by cohomology classes}
    \begin{equation}
        \dcg
    \end{equation}
        which captures the essence of $q$-form symmetries in the context of bordisms. We explore briefly all the key components of this category: objects, morphisms, monoidal structure, and braiding, shedding light on their significance and properties.

    \paragraph{Objects} For the objects of this category, we consider pairs $(\Si,b)$ of closed oriented smooth $(d-1)$-dimensional manifolds $\Si$ with a choice of $(d-q-2)$-skeleton $\sk\subset \Si$ and a background field  $b\in \coh[q+1]{\Si}{\hSi}$, where $\hat{\Si}\equiv \Si\setminus \sk$ is the complement of the skeleton.  In Appendix \ref{secA1} we collect various properties of the relative homology and cohomology groups involving $\Sigma$ and $\sk{}$ that we will use throughout the paper.

    \paragraph{Morphisms} 
    Morphisms $(\Sigma_-,b_-)\rightarrow (\Sigma_+,b_+)$ are pairs $(M, B)$, where $M$ is a morphism $\Sigma_-\rightarrow \Sigma_+$ in the category of ordinary smooth oriented bordisms (equipped with the embeddings $\Sigma_\pm \stackrel{i_\pm}{\hookrightarrow} \partial M\subset M$), and $B$ is an element in the cohomology group $\coh[q+1]{M}{\hdm}$ subject to the following condition. One can consider the inclusions $i_\pm$ as the maps between pairs of topological spaces: $i_\pm:\,(\Sigma_\pm, \hSi_\pm) \rightarrow (M, \hdm)$. They induce the pullbacks $i_\pm^*: \coh[q+1]{M}{\hdm} \rightarrow \coh[q+1]{\Sigma_\pm}{\hSi_\pm}$. We require that $i_\pm^*(B) = b_\pm$. Physically this means that the fields $b_\pm$ on the boundaries $\Sigma_\pm$ are the restrictions of the field $B$ in the bulk manifold $M$.
    To fully characterize these morphisms, we also adjust the usual bordism equivalence relation incorporating the fields. We identify $(M, B) \sim (M', B')$ when there exists a diffeomorphism $\psi: M \rightarrow M'$, identical on the boundary, that induces an isomorphism of the cohomology groups $\coh[q+1]{M}{\hdm} \cong \coh[q+1]{M'}{\hdm'}$ such that $\psi^*(B') = B$.

\begin{figure}
    \centering
    \includegraphics[width=0.6\textwidth]{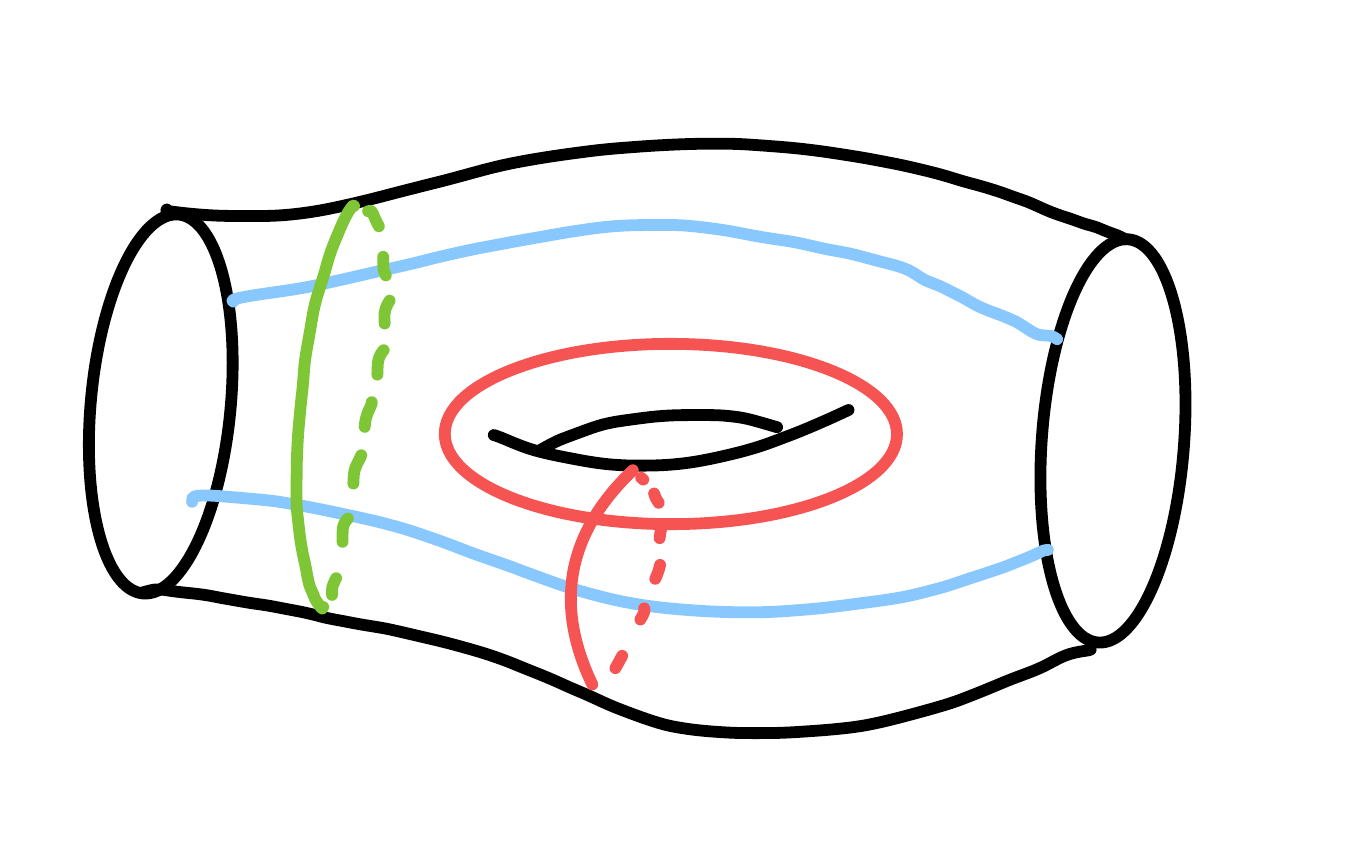}
    \caption{An example of bordism $(M,B):(\Sigma_-,b_-)\rightarrow (\Sigma_+,b_+)$ in $\mathbf{Bord}_{2}^{0,G}$. $M$ is a surface which provides a bordism between two circles $\Sigma_{\pm}\cong S^1$. Here, for each circle its $0$-skeleton $\sk_\pm$ consists of two points. The colored lines are defects representing the class $B\in H^{1}(M,\hat{\partial M})\cong H_{1}(M,\sk_+\sqcup \sk_-)$. The cyan lines end on $\sk_\pm\subset \Sigma_\pm$ while red and green lines lie completely in the bulk manifold $M$. Their boundaries represent classes $b_\pm\in H^1(\Sigma_\pm,\hat{\Sigma}_\pm)\cong H_0(\sk_\pm)$.  }
    \label{fig:1defrel}
\end{figure}

    Geometrically such decorated bordisms can be understood as follows. By Poincaré-Alexander-Lefschetz duality, we have $\coh[q+1]{\Si_\pm}{\hSi_\pm}\cong H_{d-q-2}(\sk_\pm)$ and $H^{q+1}(M,\hat{\partial M})\cong H_{d-q-1}(M,\sk_+\sqcup \sk_-)$. In terms of the dual homology group, the class $B$ is represented by a cycle in $M$ relative to $\sk_\pm$. It can be understood as a $(d-q-1)$-dimensional defect in $M$ possibly ending on the skeletons $\sk_\pm\subset \partial M$. Its boundary is a $(d-q-2)$-dimensional chain in $\sk_\pm\subset \Sigma_{\pm}$ representing $b_\pm$. An example with $d=2$ and $q=0$ can be seen in Figure \ref{fig:1defrel}. Note that there is a surjective homomorphism $H^{q+1}(\Sigma,\hat{\Sigma})\cong H_{d-q-2}(\sk)\rightarrow H^{q+1}(\Sigma)\cong H_{d-q-2}(\Sigma)$, so there is ''no loss of physical information'' by restricting the defects on the boundary to the skeleton.

    Next, we need to define the \textit{composition} of the morphisms. Let $(M_-,B_-):(\Sigma_-,b_-)\rightarrow (\sigma,b)$ and   $(M_+,B_+):(\sigma,b)\rightarrow (\Sigma_+,b_+)$. We then define the composition
     \begin{equation}
        (M,B) = (M_+,B_+)\circ(M_-,B_-):\;(\Sigma_-,b_-)\rightarrow (\Sigma_+,b_+)
    \end{equation}
    as follows. For the first entry $M$, the composition follows the standard composition of bordisms: $M=M_+\circ M_-$, i.e. by gluing them along $\sigma$. For the second entry, we must clarify the relationship between the field configuration $B \in \coh[q+1]{M}{\hdm}$ and the combined field configurations $B_- \oplus B_+ \in \coh[q+1]{M_-}{\hsi_-} \oplus \coh[q+1]{M_+}{\hsi_+}$ where $\sigma_\pm = \Sigma_\pm \sqcup \sigma$. To understand this, we utilize the relative version of the Mayer-Vietoris sequence decomposing $(M, \Si)$ into $(M_-, \sigma_-)$ and $(M_+, \sigma_+)$. It is important to note that the intersection of spaces leads to $(\sigma, \hsi)$, which yields non-trivial cohomology groups, because we have $\coh[k]{\sigma}{\hsi} = 0$ only up to $k=q$. The details of the calculations can be again found in Appendix \ref{secA1}. Thanks to the gluing condition asking for the fields to be identical on the cut $\sigma$ we have a unique identification between $\overline{B} \in \coh[q+1]{M}{\hSi}$, where $\Sigma:=\Sigma_+\sqcup \sigma \sqcup \Sigma_-$ and $(B_+,B_-) \in \coh[q+1]{M_+}{\hsi_+} \bigoplus \coh[q+1]{M_}{\hsi_-}$. 
    What we are left to do is understand how to relate $\overline{B} \in \coh[q+1]{M}{\hSi}$ to our physical configuration $B \in \coh[q+1]{M}{\hdm}$. In Appendix \ref{secA1} we show that there is a surjective map
    \begin{equation}
        \eta : {\coh[q+1]{M}{\hSi}} \longrightarrow  \coh[q+1]{M}{\hdm}.
    \end{equation}
    We then take $B=\eta(\overline{B})$. 
    The composition is associative since it is associative for bordisms, and the direct sum of groups is also associative. Geometrically the composition corresponds to gluing the bordisms with defects along the common boundary $\sigma$ and then ``unrestricting'' the defects from the skeleton in the cut $\sigma$.

    Next we need to define the \textit{identity} morphism $(\Sigma,b)\rightarrow (\Sigma,b)$, where both copies of $\Sigma$ are assumed to have the same chosen $(d-q-2)$-skeleton $\sk\subset \Sigma$. It consists of two components: the identity morphism $\Si \times I$ (where $I:=[0,1]$ is an interval) in the usual bordism category and an element $B \in {\coh[q+1]{\Si \times I}{\hSi_+\sqcup \hSi_-}}$, where $\hat{\Sigma}_\pm$ are two copies of $\hat{\Sigma}$. To describe $B$ explicitly, we will study the group
    \begin{equation}
        \coh[k]{\Si \times I}{\hSi_+\sqcup \hSi_-} \cong \homo[d-k]{\Si \times I}{\sk_+\sqcup \sk_-}
    \end{equation}
    for $k \leq q+1$, where $\sk_\pm \subset \Sigma_\pm$ are two copies of the same skeleton $\sk\subset \Sigma$. To do so, we examine the long exact sequence for the triple
    \begin{alignat*}{2}
        {\sk}_+\sqcup {\sk}_- &\subset {\sk} \times I &&\subset \Si \times I.
    \end{alignat*}
    An example of identity morphism for $d=2$ and $q=0$ is depicted in Figure \ref{fig:IdTri}.
    \begin{figure}
        \centering
        \includegraphics[width=0.5\textwidth]{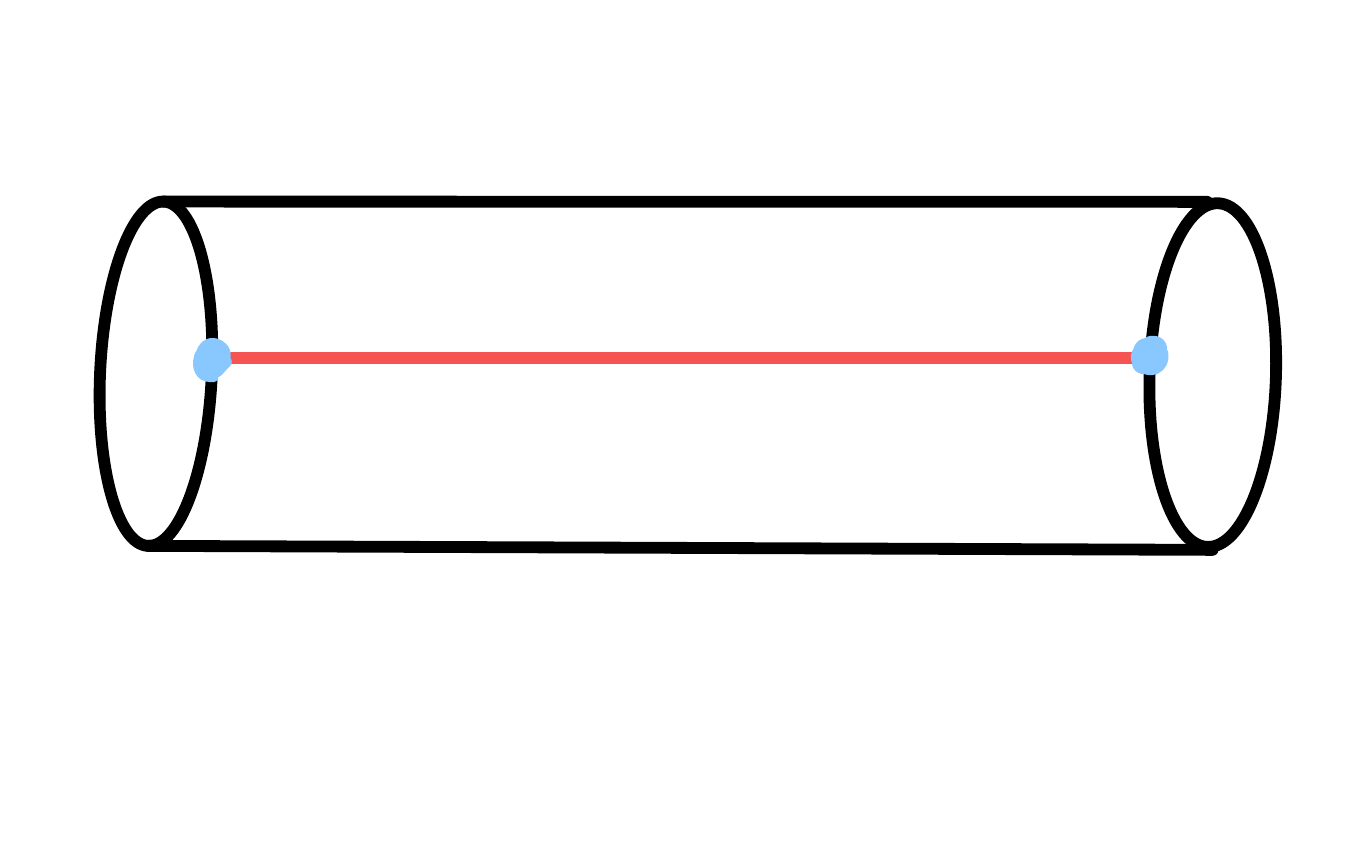}
        \caption{An example of a triple $({\sk}_+\sqcup \sk_- \subset \sk \times I \subset \Si \times I)$ for a 0 form symmetry in 2 dimensions. The cyan dots represent 1-point sets ${\sk}_\pm\subset \Sigma_\pm \cong S^1$, the red line represents ${\sk} \times I$, and the black cylinder represents $\Si \times I$. The red line also represents the ``identity'' field configuration preserving the boundary field configurations represented by cyan dots.}
        \label{fig:IdTri}
    \end{figure} 
    In Appendix \ref{secA1} we show all the details that allow us to see how the long exact sequence of this triple reduces to the short exact sequence
    \begin{center}
    \begin{tikzcd}
        0 \arrow[r]
        & \homo[d-q-1]{{\sk} \times I}{\sk_+\sqcup \sk_-} \arrow[r]
        & \homo[d-q-1]{\Si \times I}{{\sk}_+\sqcup \sk_-} \arrow[r] 
        & H_{d-q-1}(\Si,\sk) \arrow[r] 
        & 0.\\
    \end{tikzcd}
    \end{center}
    This sequence splits (in a canonical way, if one assumes a preference of one cylinder end to the other) allowing us to say that $ \coh[q+1]{\Si \times I}{\hSi_\pm} \cong  H^{q+1}(\Si,\hSi) \bigoplus H^{q}(\hSi)$ which, in terms of homology, is telling us that defects can be decomposed into a component "parallel" to the interval $I$ and a component "perpendicular" to it: $\homo[d-q-1]{\Si \times I}{{\sk}_+\sqcup \sk_-} \cong \homo[d-q-1]{{\sk} \times I}{{\sk}_+\sqcup \sk_-} \bigoplus H_{d-q-1}(\Si,\sk)$. As we will see later, physically, the perpendicular component represents the action of the $q$-form symmetry on the Hilbert space, while the parallel one determines in which twisted sector we are. In pictorial terms this splitting can be seen for 1-form symmetries in 3 dimensions in Figure \ref{fig:1Split3d}.
    \begin{figure}
        \centering
        \includegraphics[width=0.5\textwidth]{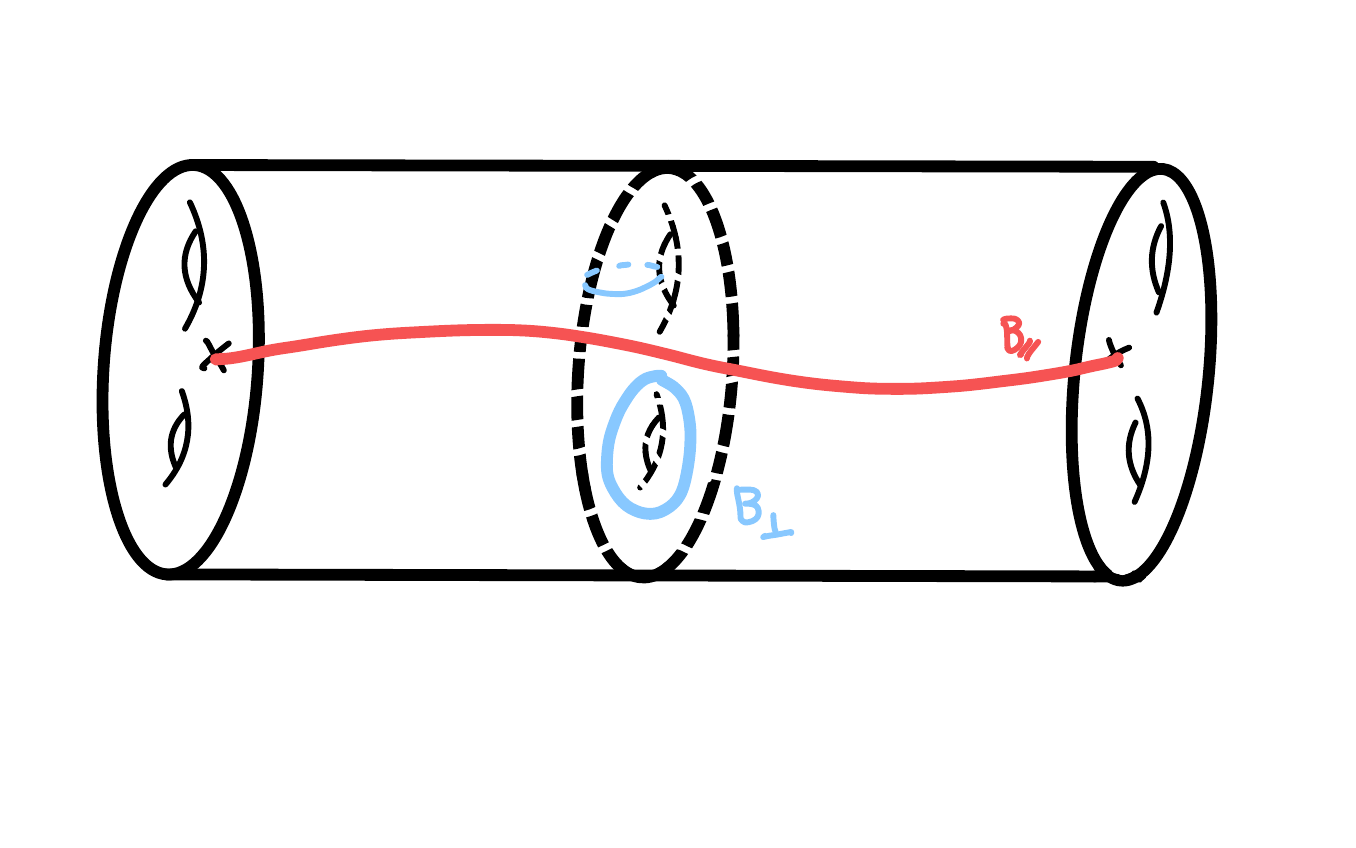}
        \caption{The case of a 1-form symmetry in 3 dimensions, which corresponds to 1-dimensional topological defects that can be decomposed into a parallel defect (red) and a perpendicular one (cyan).}
        \label{fig:1Split3d}
    \end{figure}     
    For the identity morphism, we just have to consider parallel defects i.e. $B = b \oplus 0$. Consequently, the identity morphism becomes $(\Si \times I, b \oplus 0)$ where clearly $b \oplus 0$ restricts to $b$ on both boundaries of the cylinder. The fact that it is indeed an identity morphism follows immediately from the explicit definition of the composition of morphisms above.
    In this way we have shown how this refined bordism category is a well-defined category. 
    
    Next, we show how it can be considered a symmetric monoidal category.
    For \textit{monoidality}, recalling the property of cohomology groups being split into direct sums when applied to disjoint unions, we can readily establish a satisfactory definition of the tensor product:
    \begin{equation}
        (\Sigma_1,b_1)\otimes (\Sigma_2,b_2) = (\Sigma_1 \sqcup \Sigma_2,b_1\oplus b_2).
    \end{equation}
    The identity object for the tensor product is $(\emptyset,0)$. The associativity of the tensor product follows from the associativity for both the entries of the pair. The construction of the tensor product for morphisms follows the same principles as for objects. 
    
    We then establish the \textit{braiding structure} through the cylinder construction. That is, we consider
\begin{equation}
   \beta_{(\Sigma_1,b_1)(\Sigma_2,b_2)} =((\Sigma_1 \sqcup \Sigma_2) \times I,(b_1 \oplus b_2)\oplus (0\oplus 0))
\end{equation}
as a morphism
\begin{equation}
   ( \Sigma_1 \sqcup \Sigma_2,b_1 \oplus b_2)\longrightarrow ( \Sigma_2 \sqcup \Sigma_1,b_2 \oplus b_1).
\end{equation}
    This construction satisfies the property required for a symmetric category: $\beta_{(\Sigma_1,b_1)(\Sigma_2,b_2)} = \beta_{(\Sigma_2,b_2)(\Sigma_1,b_1)}^{-1}$.\bigskip
    
    In this way, we have fully described $\dcg$ as a symmetric monoidal category. We now use this in our definition of TQFT with a global finite $q$-form symmetry.

    \subsection{TQFT with global symmetry}
    We can now proceed to discuss the TQFT as a functor on the refined bordism category $\dcg$:
    \begin{definition}
        A closed oriented TQFT $\mathcal{Z}$ with a global $q$-form symmetry $G$ is a symmetric monoidal functor from the bordism category of manifolds decorated by cohomology classes to the category of complex\footnote{For the purpose of defining a TQFT with a global symmetry, any field can be considered. However, for the purpose of defining the gauging procedure later we need to restrict ourselves to $\C$.} vector spaces:
        \begin{equation}
            \mathcal{Z}: \dcg \longrightarrow \bf{Vect}_{\mathbb{C}}.
        \end{equation}
    \end{definition} 
    We shall now proceed to shed light on the details of this functor. We denote the value of $\mathcal{Z}$ on objects $(\Sigma,b)$ as $\Hi(\Sigma,b):=\mathcal{Z}(\Sigma,b)$, and its value on morphisms $(M,B)$ as $Z(M,B):=\mathcal{Z}(M,B)$ to be coherent with the common physical notations. The vector space $\Hi(\Si,b)$ has a physical meaning of the Hilbert space on the $(d-1)$-dimensional spatial slice $\Si$ in the presence of the background gauge field $b\in \coh[q+1]{\Si}{\hSi}$. The value on the bordism $(M,B):(\Sigma_-,b_-)\rightarrow (\Sigma_+,b_+)$ has the meaning of the evolution operator $Z(M,B):\Hi(\Si_-,b_-)\rightarrow \Hi(\Si_+,b_+)$ in the presence of the background gauge field $B$ in the spacetime $M$ that restricts to $b_+$ and $b_-$ on the boundaries. When the manifold is closed, $Z(M,B)\in \C$ has the meaning of the partition function of the theory in the presence of the background gauge field $B\in H^q(M)$.

    The conditions on $\mathcal{Z}$ being functorial can be explicitly stated as follows: 
    \begin{align}
        Z&(M,B) = Z(M,\eta(\overline{B}))=Z(M_+,B_+)\circ Z(M_-,B_-),\\
        Z&(\Si \times I, b \oplus 0) = \text{Id}_{\Hi(\Si,b)}.
        \label{ungauged-functoriality}
    \end{align}
    That is $\mathcal{Z}$ preserves composition and identity.
Physically composition corresponds to the composition of the evolution operators between the Hilbert spaces at different spatial slices. 
    
    The condition of being monoidal in particular implies $\Hi(\Sigma_1,b_1)\otimes \Hi(\Sigma_2,b_2) = \Hi(\Sigma_1 \sqcup \Sigma_2, b_1 \oplus b_2)$, which physically corresponds to the fact that Hilbert space of two non-interacting systems is the tensor product of the respective Hilbert spaces.
    The condition of being symmetric means that it preserves the \textit{braiding structure}: $Z(\beta_{(\Sigma_1,b_1),(\Sigma_2,b_2)}) = \beta_{\Hi(\Sigma_1,b_1),\Hi(\Sigma_2,b_2)}$, where $\beta$ for the vector spaces is given by the permutation in the tensor product.

 Note that physically we expect the Hilbert space $\Hi(\Sigma,b)$ to depend on the background field valued in $H^{q+1}(\Sigma)$, not in $\coh[q+1]{\Si}{\hSi}$, which is, in general, a larger group, in the sense that there is a non-trivial surjective map $g:\coh[q+1]{\Si}{\hSi}\rightarrow H^{q+1}(\Si)$ (see Appendix \ref{secA1}). However, as we will see below, all $\Hi(\Sigma,b)$ for different choices of the skeleton $\sk\subset \Sigma$ and $b$, such that $g(b)\in H^{q+1}(\Sigma)$ is fixed, are isomorphic to each other. Namely, there is an isomorphism  $\Hi(\Si,b)\cong \Hi(\Si',b')$ if $g(b)=g'(b')$, where $\Si'$ is a copy of $\Si$ with a possibly different skeleton $\sk'$ and $g'$ is the analogue of $g$ for $\Sigma'$: $g':H^{q+1}(\Sigma,\hat{\Sigma}')\rightarrow H^{q+1}(\Sigma)$. The isomorphism is provided by the value of the TQFT on the cylinder $\Sigma\times I$ with the respective skeletons at different ends and a field $B$ that restricts to $b$ and $b'$ at those ends: 
\begin{proposition} 
    There exists an element $B \in H^{q+1}(\Si \times I,\hat\Sigma \sqcup \hat\Sigma')$ such that $i_-^*(B)=b$ and $i_+^*(B)=b'$, where $i_\pm$ are the identifications $i_-:\Sigma\rightarrow \Sigma\times \{0\}$ and $i_+:\Sigma'\rightarrow \Sigma\times \{1\}$ if and only if $g(b)=g'(b')$ (where $g$ and $g'$ are as defined above).
    \label{prop:skeleton-change}
\end{proposition}
\begin{proof}
Recall that by Poincar\'e-Alexander-Lefschetz duality we have $H^{k}(\Si \times I,\hat\Sigma \sqcup \hat\Sigma')\cong H_{d-k}(\Sigma\times I,\sk\sqcup\sk')$, $H^k(\Sigma,\hat{\Sigma})\cong H_{d-k}(\sk)$, and $H^k(\Sigma',\hat{\Sigma}')\cong H_{d-k}(\sk')$.
 Let us look at the long exact sequence of homology groups for the pair $({\sk} \sqcup {\sk}',\Si \times I)$:

\begin{equation}
\begin{tikzcd}
 \ldots \ar[r] &  H_{d-q-1}(\sk)\oplus H_{d-q-1} (\sk') \ar[r] &  H_{d-q-1}(\Sigma\times I) \ar[r] &   H_{d-q-1}(\Sigma\times I,\sk\sqcup \sk') \ar[r,"\partial"] &  {}  \\
& 0 \ar[u,equal] &   H^{q}(\Sigma) \ar[u,equal] &  H^{q+1}(\Si \times I,\hat\Sigma \sqcup \hat\Sigma') \ar[u,equal]   &  
\\
 \ar[r] &  H_{d-q-2}(\sk)\oplus H_{d-q-2} (\sk')\ar[r,"f"] & H_{d-q-2}(\Sigma\times I) \ar[r] &   H_{d-q-2}(\Sigma\times I,\sk\sqcup \sk')  \ar[r] & \ldots  \\
& H^{q+1}(\Sigma,\hat{\Sigma})\oplus  H^{q+1}(\Sigma',\hat{\Sigma}') \ar[u,equal] &   H^{q+1}(\Sigma) \ar[u,equal] &      &  
\end{tikzcd}
\end{equation}
First note that in terms of cohomology groups, the connecting homomorphism is given by $\partial=(-i_-^*)\oplus i_+^*$. The two components of the map $f$, in terms on homology groups, are induced by inclusions $\sk,\sk'\hookrightarrow \Sigma \times I$, which then can be composed with the isomorphism induced by the deformation retraction $\Sigma\times I\rightarrow \Sigma$. On the other hand, the maps $g,g'$, under Poincar\'e-Lefschetz duality can be identified with the maps $H_{d-q-2}(\sk)\rightarrow H_{d-q-2}(\Sigma)$ and $H_{d-q-2}(\sk')\rightarrow H_{d-q-2}(\Sigma')$ induced by inclusions of the skeletons directly into $\Sigma$ and $\Sigma'$. From this it follows that in terms of dual cohomology groups we can identify $f=g+g'$. The statement of the proposition then follows from the exactness of the sequence.
\end{proof}
The considered cylinder $(\Sigma\times I,B)$ is an invertible morphism in our decorated category, as it can be composed with its copy with reversed orientation to provide the identity morphism considered earlier. The value of the TQFT on it then provides an isomorphism:
\begin{equation}
    Z(\Sigma\times I,B):\Hi(\Si,b) \stackrel{\cong}{\longrightarrow} \Hi(\Si',b').
    \label{iso-global-tqft}
\end{equation}
It follows that $\Hi(\Sigma,b)$, up to an isomorphism only depends on $g(b)$, but not on $b$ itself or the choice of the $(d-q-2)$-skeleton $\sk\subset \Sigma$ (and therefore also on the cell-decomposition of $\Sigma$ used to define the skeleton). The isomorphism $\Hi(\Si',b) \cong \Hi(\Si',b')$, however, is not canonical in general, as there might be different $B$ that restrict to $b$ and $b'$ at the ends of the cylinder $\Sigma\times I$. As can be seen from the long exact sequence in the proof of Proposition \ref{prop:skeleton-change}, any two such choices must differ by an element in the image of the injective map $H^q(\Sigma)\rightarrow H^{q+1}(\Sigma\times I,\hat\Sigma\sqcup \hat\Sigma')$.

    From the definition of the TQFT it follows that the vector spaces $\Hi(\Sigma,b)$ naturally form a representation $\rho_{\Sigma,b}$ of the group $H^q(\Sigma)\cong H_{d-q-1}(\Sigma)$. The action of the group is provided by the value of the TQFT on the cylinder $\Sigma\times I$ with the field $B\in \coh[q+1]{\Si \times I}{\hSi_\pm} \cong  H^{q+1}(\Si,\hSi) \bigoplus H^{q}(\hSi)$ of the form $B=b\oplus B_\perp$. In order for the cylinder to provide an endomorphism of $(\Sigma,b)$, $B_\perp$ must be inside the image of the injective map $\iota:H^q(\Sigma)\rightarrow H^q(\hat{\Sigma})$ (see Appendix \ref{secA1}).  That is, we have
    \begin{equation}
        \rho_{\Sigma,b}(\beta)=Z(\Sigma\times I,b\oplus \iota(\beta)):\Hi(\Sigma,b)\rightarrow \Hi(\Sigma,b),
    \qquad
    \beta\in  H^{q}(\Sigma)\cong H_{d-q-1}(\Sigma).
    \end{equation}
Physically the action is given by the $(d-q-1)$-dimensional defects parallel to the spatial slice $\Sigma$. 

The cylinder providing the isomorphism (\ref{iso-global-tqft}) can be composed with the cylinder providing the action of $\beta\in H^q(\Sigma)$ on $\mathcal{H}(\Sigma,b)$ or $\mathcal{H}(\Sigma,b')$. The result of the composition is again a cylinder of the form (\ref{iso-global-tqft}) but with $B$ shifted by the image of $\beta$ under the injection $H^q(\Sigma)\rightarrow H^{q+1}(\Sigma\times I,\hat\Sigma\sqcup \hat\Sigma')$. It then follows that the isomorphism (\ref{iso-global-tqft}) is equivariant with respect to the action of $H^q(\Sigma)$, and, moreover, that any two isomorphisms (\ref{iso-global-tqft}) differ by $\rho_{\Sigma,b}(\beta)$ or, equivalently $\rho_{\Sigma',b'}(\beta)$ for some $\beta\in H^q(\Sigma)$. This will be important for the construction of the gauged TQFT functor. In Section \ref{sec4} we will argue that for the Hilbert spaces with the same skeleton, there is a canonical isomorphism $\Hi(\Sigma,b)\cong \Hi(\Sigma,b')$ for $g(b)=g(b')$, defined by considering the direct sum decomposition into irreducible representations.

\section{Gauging a \texorpdfstring{$q$}{q}-form symmetry}\label{sec3}

We are now interested in gauging the $q$-form symmetry $G$. In this section, starting from a functor $\mathcal{Z}:\mathbf{Bord}_{d}^{q,G}\rightarrow  \mathbf{Vect}_\C$ describing a TQFT with a global $q$-form symmetry we will construct a new ordinary TQFT functor $\mathcal{Z}_\text{g}:\mathbf{Bord}_{d}\rightarrow  \mathbf{Vect}_\C$ defined on the usual, undecorated, bordism category.  To do so, will first give the ``naive'' values of this functor on the objects and morphisms, which we will then adjust to provide the correct ones that satisfy the functoriality properties. In Section \ref{sec4} we will show how to extend this functor to a functor $\mathbf{Bord}_{d}^{d-q-2,G^*}\rightarrow  \mathbf{Vect}_\C$ where $G^*:=\Hom(G,\R/\Z)$. That will provide a formal description of dual global symmetry, see \cite{Gaiotto,Bhardwaj:2017xup}, in the setting considered in this paper.

The basic idea is to ``sum over cohomology classes'' decorating the manifolds in the bordism category. Physically, we want to sum only over gauge inequivalent field configurations. As was mentioned previously, for a $(d-1)$-manifold $\Sigma$ configurations $b\in H^{q+1}(\Sigma,\hat{\Sigma})$ with the same value of $g(b)\in H^{q+1}(\Sigma)$ should be considered equivalent. Naively, then one can then define the Hilbert spaces of the gauged TQFT as follows:
\begin{equation}
\label{eq:naive}
    \Hi^\text{n}_\text{g}(\Si) := \bigoplus_{\tilde{b} \in H^{q+1}(\Si)} \Hi(\Si,b)
\end{equation}
where, in each term, $b$ is some element of $H^q(\Sigma)$ such that $g(b)=\tilde{b}$. Note that to define  $\Hi^\text{n}_\text{g}(\Si)$ here for each $\Sigma$ we need to choose a particular skeleton $\sk\subset \Sigma$, and moreover specify a choice of the preimage $b$ for each element of $H^q(\Sigma,\hat\Sigma)$. The second part more formally can be stated as a choice of the right-inverse function (not necessarily a homomorphism) to the surjection $g$. The dependence on such choices will be addressed later.

We then define the value $Z^\text{n}_\text{g}(M)$  on the bordism $M:\Sigma_-\rightarrow \Sigma_+$ by the following expression for its components:
\begin{equation}
   Z_{\text{g},b_+,b_-}^\text{n}(M) = c(M) \sum_{\substack{B \in \coh[q+1]{M}{\hdm} \\i_\pm^*(B)=b_\pm}} Z(M,B):\qquad \Hi(\Sigma_-,b_-)\rightarrow \Hi(\Sigma_+,b_+) 
   \label{naive-bordism-value}
\end{equation}
where $\Hi(\Sigma_\pm,b_\pm)$ is a component in the direct sum decomposition of $\Hi^\tn_\tg(\Sigma_\pm)$ and $c(M)$ is a coefficient that we will specify shortly and that will be needed in order for the formula to be consistent with the composition of bordisms. Consider now $(M,B)$ as a composition of two bordisms $(M_\pm,B_\pm)$ with the same conventions as in Section \ref{sec2} (see also Appendix \ref{secA1}). The boundary condition for the fields are specified as $i^*_\pm(B) =i^*_\pm(\overline{B})=i^*_\pm(B_\pm) = b_\pm$, $i^*(\overline{B})=b_\sigma$ and $i^*_{\pm,\sigma}(B_\pm)=b_{\sigma}$ for the inclusions $i_\pm: \Sigma_\pm \rightarrow M_\pm \subset M$, $i:\sigma\rightarrow M$, $i_{\pm,\sigma}:\sigma\rightarrow M_\pm$. We begin now from the definition of $Z_{\tg,b_+,b_-}^\tn(M)$ and will rewrite it in terms of the composition of the two partition functions $Z_{\tg,b_+,b_\sigma}^\tn(M_+)$ and $Z_{\tg,b_\sigma,b_-}^\tn(M_+)$: 
\begin{align*}
    Z_{\tg,b_+,b_-}^\tn(M) &= c(M) \sum_{\substack{B \in \coh[q+1]{M}{\hdm} \\ i_\pm^*(B)=b_\pm}} Z(M,B)\\
    &= c(M)\frac{1}{|\text{Ker}\eta|} \sum_{\substack{\overline{B} \in \coh[q+1]{M}{\hSi} \\ i^*_\pm(\overline{B})=b_\pm }} Z(M,\eta(\Bar{B}))\\
    &= c(M) \frac{1}{|\text{Ker}\eta|} \sum_{b_\sigma \in \coh[q+1]{\sigma}{\hat\sigma}} \sum_{\substack{B_\pm \in \coh[q+1]{M_\pm}{\hsi_i} \\ i^*_\pm(B_\pm)=b_\pm,\; i^*_{\pm,\sigma}(B_\pm)=b_\sigma }} Z(M_+,B_+)\circ Z(M_-,B_-) \\
    &= \frac{c(M)}{c(M_+)c(M_-)}\frac{|\text{Im}\Phi|}{|\text{Ker}\eta|} \sum_{{\tilde{b}}_\sigma\in H^{q+1}(\sigma)} Z_{\tg,b_{+},b_\sigma}^\tn(M_+)\circ Z_{\tg,b_{\sigma},b_-}^\tn(M_-)\\
    &= \sum_{{\tilde{b}}_\sigma\in H^{q+1}(\sigma)}  Z_{\tg,b_{+},b_\sigma}^\tn(M_+)\circ Z_{\tg,b_{\sigma},b_{-}}^\tn(M_-).
\end{align*}
The second equality is due to the fact that using $\coh[q+1]{M}{\hSi}$ instead of $\coh[q+1]{M}{\hdm}$ only yields an overcounting in the sum since $Z(M,\eta(\Bar{B}))=Z(M,\eta(\Bar{B} + c))$ for any $c \in  \text{Ker}\eta$. In the third equality, we used the functoriality (\ref{ungauged-functoriality}) and split the sum over $\overline{B}$ into the sum over $B_\pm$ with fixed boundary conditions and the sum over the intermediate field value $b_\sigma$ on the cut $\sigma$. 
Next step is to notice that the composition $Z_{\tg,b_{+,\sigma}}^\tn(M_+)\circ Z_{\tg,b_{\sigma,-}}^\tn(M_-)$ depends only on $B$. Therefore it depends on $b_\sigma$ only through its image $g_\sigma(b_\sigma)$, $g_\sigma:H^{q+1}(\sigma,\hat\sigma)\rightarrow H^{q+1}(\sigma)$, see Appendix \ref{secA1}. By looking again at the long exact sequence for the pair $(\sigma,\hat{\sigma})$ we notice $H^{q+1}(\sigma) = \bigslant{\coh[q+1]{\sigma}{\hat\sigma}}{\text{Im}\Phi}$ where $\Phi:H^{q}(\hat\sigma)\rightarrow H^{q+1}(\sigma,\hat\sigma)$ and therefore $\text{Im}\Phi = \frac{|H^q(\hat{\sigma})|}{|H^q(\sigma)|}$. Thus the sum over $H^{q+1}(\sigma,\hat\sigma)$ can be rewritten as the sum over $H^{q+1}(\sigma)$ times the size of the image of $\Phi$. By taking $b_\sigma$ to be a preimage of ${\tilde{b}}_\sigma$ according to the choice made in the definition of $\Hi_\tg^\tn(\sigma)$, one can then rewrite the expression using the definition of partition function (\ref{naive-bordism-value}) for $M_\pm$. This gives the fourth equality. By defining the coefficient $c(M)$ as\footnote{Here we make a choice which is symmetric under the exchange of the in- and out-boundaries $\Sigma_-$ and $\Sigma_+$. More generally, one could use
\begin{equation}
    c(M)=\prod_{i=0}^q \left( |\coh[i]{M}{\hdm} | |H^{i}(\Si_+) |^{s} |H^{i}(\Si_-) |^{1-s} \right )^{(-1)^{q-i+1}} 
\end{equation}
for some $s$ (the natural choices are $s=0,1,\frac{1}{2}$). The resulting functor would be the same up to a natural isomorphism.
} 
\begin{equation}
    c(M)=\prod_{i=0}^q \left( |\coh[i]{M}{\hdm} | |H^{i}(\Si_+) |^{\frac{1}{2}} |H^{i}(\Si_-) |^{\frac{1}{2}} \right )^{(-1)^{q-i+1}} 
 \end{equation}
we finally obtain the last equality showing that $Z^\tn_\tg(M)=Z^\tn_\tg(M_+)\circ Z^\tn_\tg(M_+)$ and therefore that he composition rule is preserved. That is, we have a commutative diagram
\begin{center}
    \begin{tikzcd}
        \Hi^\text{n}_\text{g}(\Si_-) \arrow[rr,"Z^\tn_\tg(M_-)"] \arrow[rrrr, bend right,"Z^\tn_\tg(M)"] & {} & \Hi^\text{n}_\text{g}(\sigma) \arrow[rr,"Z^\tn_\tg(M_+)"] & {} 
        & \Hi^\text{n}_\text{g}(\Si_+)
    \end{tikzcd}
\end{center}
In pictorial terms, this composition can be seen in Figure \ref{fig:Comp-1}

\begin{figure}
    \centering
    \includegraphics[width=0.5\textwidth]{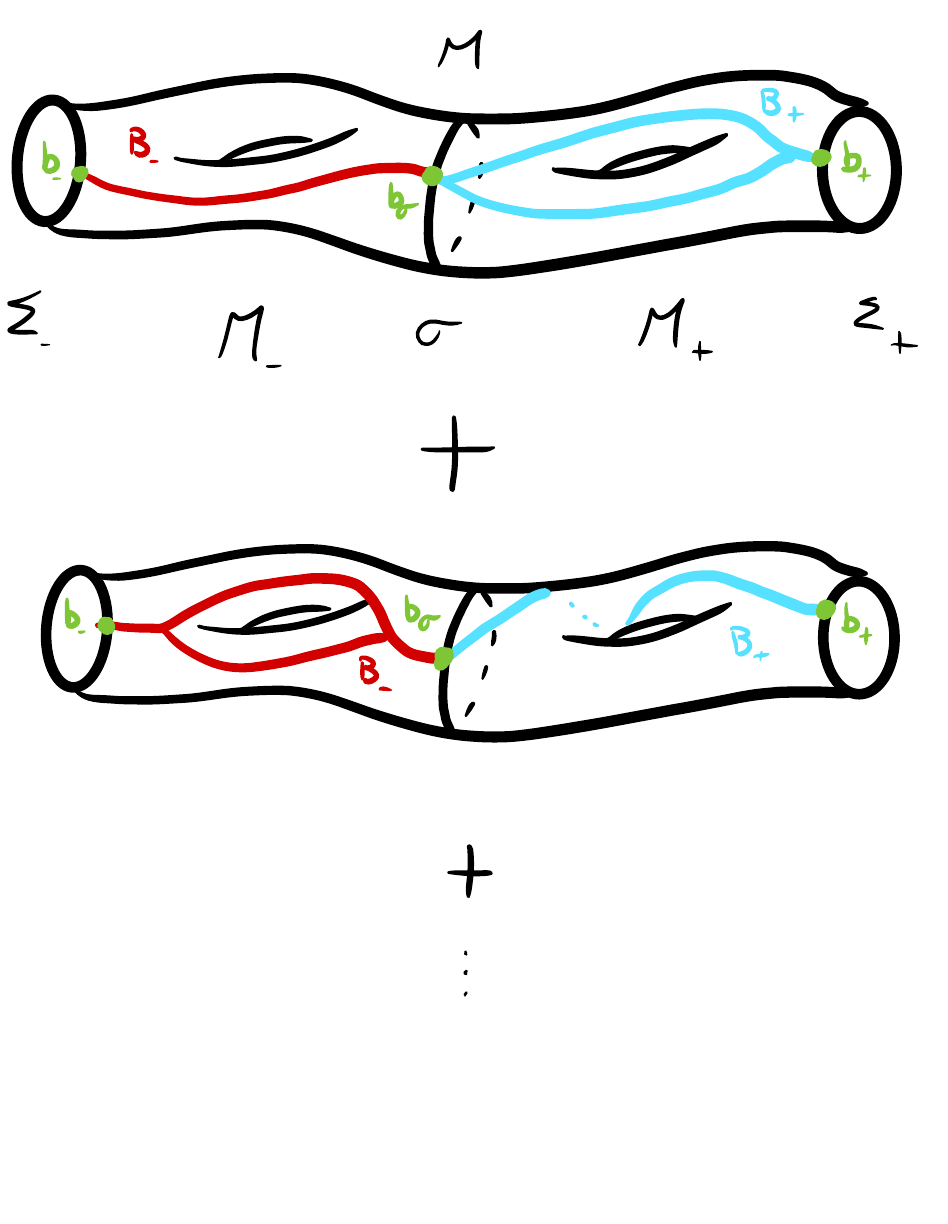}
    \caption{Composition of the gauged partition functions. We consider all possible values of fields $B_\pm$ inside $M_\pm$ and $b_\pm$  on the boundaries $\Si_\pm$ (modulo the equivalence relation defined by having the same image in $H^{q+1}(\Sigma_\pm)$). In this way, the composition of the full partition functions is acting on the full Hilbert space.}
    \label{fig:Comp-1}
\end{figure}

The naivety of this definition lies in fact that the identity morphism is not preserved. The partition function on the cylinder acts as a projector instead. In order to see this let us recall that $ \coh[q+1]{\Si \times I}{\hSi_+\sqcup \hSi_-} = H^{q+1}(\Si,\hSi) \bigoplus H^{q}(\hSi)$. The first component $B_\parallel \in H^{q+1}(\Si,\hSi)$ is completely fixed by the boundary conditions, equal on both ends: $B_\parallel = b$. The summation can be done only over the image of the injection $\iota:H^q(\Sigma)\rightarrow H^{q+1}(\hat\Sigma)$. The coefficient in this case simplifies to $c(\Si \times I) = \frac{1}{|H^{q}(\Sigma)|}$ and the partition function becomes 
\begin{equation}
    Z_{\tg,b,b}^\tn(\Si \times I) = \frac{1}{|H^{q}(\Sigma) |} \sum_{\beta \in H^{q}(\Si) } Z(\Si \times I,  b \oplus \iota(\beta) )=\frac{1}{|H^{q}(\Sigma) |} \sum_{\beta \in H^{q}(\Si) } \rho_{\Sigma,b}(\beta)
\end{equation}
where $\rho_{\Sigma,b}$ is the representation of $H^{q}(\Sigma)$ on $\Hi(\Sigma,b)$ considered in Section \ref{sec2}. By the orthogonality property of characters of irreducible representations, we get that this partition function acts as a projector on the trivial representation subspace:
\begin{equation}
    Z^\tn_{\tg,b,b}(\Si \times I) = P_{\Si,b}^0.
\end{equation}
This leads to the correct definition for our gauged TQFT. Indeed we want physical states to be trivial under the gauge group action and the way to achieve this is to project to the trivial representations. With this in mind, we can proceed to define the actual gauged TQFT.\\

\begin{definition}
The gauging of a TQFT $\mathcal{Z}:\Bord_{d}^{q,G}\rightarrow \mathbf{Vect}_\C$ with a finite global $q$-form symmetry $G$ is a closed oriented TQFT given by the symmetric monoidal functor
    \begin{equation}
        \mathcal{Z}_{\tg}: \mathbf{Bord}_d \longrightarrow \mathbf{Vect}_{\mathbb{C}}.
    \end{equation}
from the $d$-dimensional bordism category to the category of vector spaces with the values on objects $\Sigma$  and morphisms $M:\Sigma_-\rightarrow \Sigma_+$ given by 
\begin{equation}
    \Hi_\tg(\Si) = P_{\Sigma}^0(\Hi_\tg^\tn(\Si) ),
    \label{gauged-hilbert-space}
\end{equation} 
\begin{equation}
    Z_\tg(M) = P_{\Sigma_+}^0 Z_\tg^\tn(M)|_{\mathrm{Im}\,P^0_{\Sigma_-}}.
\end{equation}
\end{definition}

The \textit{objects} of the bordism category are sent to Hilbert space built as the sum over sectors of all field configurations and then projected to the trivial representations only. The \textit{partition function} is the same as the naive one (since one can always realize a bordism $M:\Sigma_-\rightarrow \Sigma_+$ as a composition of itself with the cylinder $\Sigma_\pm \times I$), but we made it explicit that it acts only inside the trivial representations. By construction, it then follows that the $Z_\tg$ now preserves both the composition (since we have already shown that for the naive gauged partition function $Z^\tn_\tg$) and the identity, so that we have:
\begin{align*}
    Z_\tg(M) &= Z_\tg(M_+)\circ Z_\tg(M_-), \\
    Z_\tg(\Si \times I) &= \text{Id}_{\Hi_\tg(\Si)}.
\end{align*}
In this way, we have described how to build a TQFT with finite gauge $q$-form symmetry starting from a theory with a global $q$-form symmetry. Note that for a closed manifold $M$ (which can be considered as a bordism from the empty space to the empty space), we recover the known expression for the partition function of the finite group $q$-form gauged theory:
\begin{equation}
    Z_\tg(M)=\prod_{i=0}^{q} |H^i(M)|^{(-1)^{q-i+1}}\sum_{B\in H^{q+1}(M)} Z(M,B).
    \label{gauged-closed}
\end{equation}

As was pointed out earlier, the definition of the gauged TQFT functor involves certain choices: for each object $\Sigma$ one has to choose a $(d-q-2)$-skeleton, and then for each $\tilde{b}\in H^q(\Sigma)$ one has to choose an element $b$ such that $g(b) = \Tilde{b}$. We will argue now that TQFT functors for different skeletons are naturally isomorphic. 

In Section \ref{sec2} we have shown that there is an isomorphism $Z(\Sigma\times I,B):\Hi(\Sigma,b)\rightarrow \Hi(\Sigma',b')$ where $g(b)=g'(b')$ and $\Sigma'$ is the same as $\Sigma$ when one forgets about the skeleton. Moreover, in general, there is an ambiguity of the choice of $B$ with no canonical choice. The isomorphism is equivariant with respect to the action of $H^q(\Sigma)$, so it descends to a well-defined isomorphism between the trivial representation subspaces. Moreover, for the isomorphism restricted to the trivial representation subspaces the ambiguity disappears, because, as was argued in Section \ref{sec2}, two different isomorphisms differ by an action of $H^q(\Sigma)$. Therefore it provides a \textit{canonical} isomorphisms $P^0_{\Sigma,b}(\Hi(\Sigma,b))\cong P^0_{\Sigma',b'}(\Hi(\Sigma',b'))$. Therefore it also provides a canonical isomorphism between $\Hi_\tg(\Sigma)$ for different choices of the skeleton and different choices of  $b$ such that $\tilde{b}=g(b)$. This isomorphism can be realized as $Z^\tn_\tg(\Sigma\times I)$ but with different choices of the skeleton at the cylinder ends. It then follows that such isomorphisms provide a natural isomorphism between TQFT functors which are defined using different choices. Therefore the construction of the gauged TQFT functor is unambiguous up to a natural isomorphism.

\section{Dual global symmetry}\label{sec4}

In this section, we will show how the TQFT with gauged $q$-form symmetry $G$ (described in Section \ref{sec3}) can be refined to a TQFT with $q^*:=(d-q-2)$-form global symmetry $G^*:=\Hom(G,\R/\Z)$ in the sense described in Section \ref{sec2}. It is well-known \cite{Gaiotto,Bhardwaj:2017xup} how to write  the partition function of such refined TQFT on a closed $d$-manifold $M$ decorated by the cohomology class $A\in H^{q^*+1}(M;G^*)$:
\begin{equation}
    Z_\text{g}(M,A)=\prod_{i=0}^{q} |H^i(M;G)|^{(-1)^{q-i+1}}\sum_{B\in H^{q+1}(M;G)}Z(M,B)\,e^{-2\pi i\int_{M} A\cup B}
    \label{refined-closed-manifold}
\end{equation}
which reduces to (\ref{gauged-closed}) when $A=0$. The minus sign in the exponential is chosen for later convenience. To extend this to a TQFT functor one first needs to define the version of the Hilbert space $\mathcal{H}_\tg(\Sigma,a)$ in a non-trivial background $a\in H^{q'+1}(\Sigma,\hat{\Sigma}^*;G^*)$ where $\hat{\Sigma}^*:=\Sigma \setminus \sk^*$ and $\sk^*$ is a $(d-q^*-2)=q$-skeleton in $\Sigma$. Consider $\sk$ to be a skeleton for a \textit{triangulation} of $\sk$. We then choose the skeleton $\sk^*$ to be the skeleton of the \textit{polyhedral cell complex dual} to that triangulation (as in the standard proof of the Poincer\'e duality). Alternatively, one can take $\sk^*$ to be defined by a triangulation, while $\sk$ to be defined by the dual polyhedral decomposition. Such skeletons satisfy the property that the complement of one deformation retracts to the other (see Figure \ref{fig:Triang-1}), which provides the following homotopy equivalence:
\begin{equation}
    \Sigma\setminus \sk \simeq \sk^*,\qquad \Sigma\setminus \sk^* \simeq \sk.
    \label{dual-skeletons}
\end{equation}
More generally, $\sk$ and $\sk^*$ can be a pair of skeletons defined by more general cellular decompositions of $\Sigma$ that satisfy this property. 

\begin{figure}
    \centering
    \includegraphics[width=0.5\textwidth]{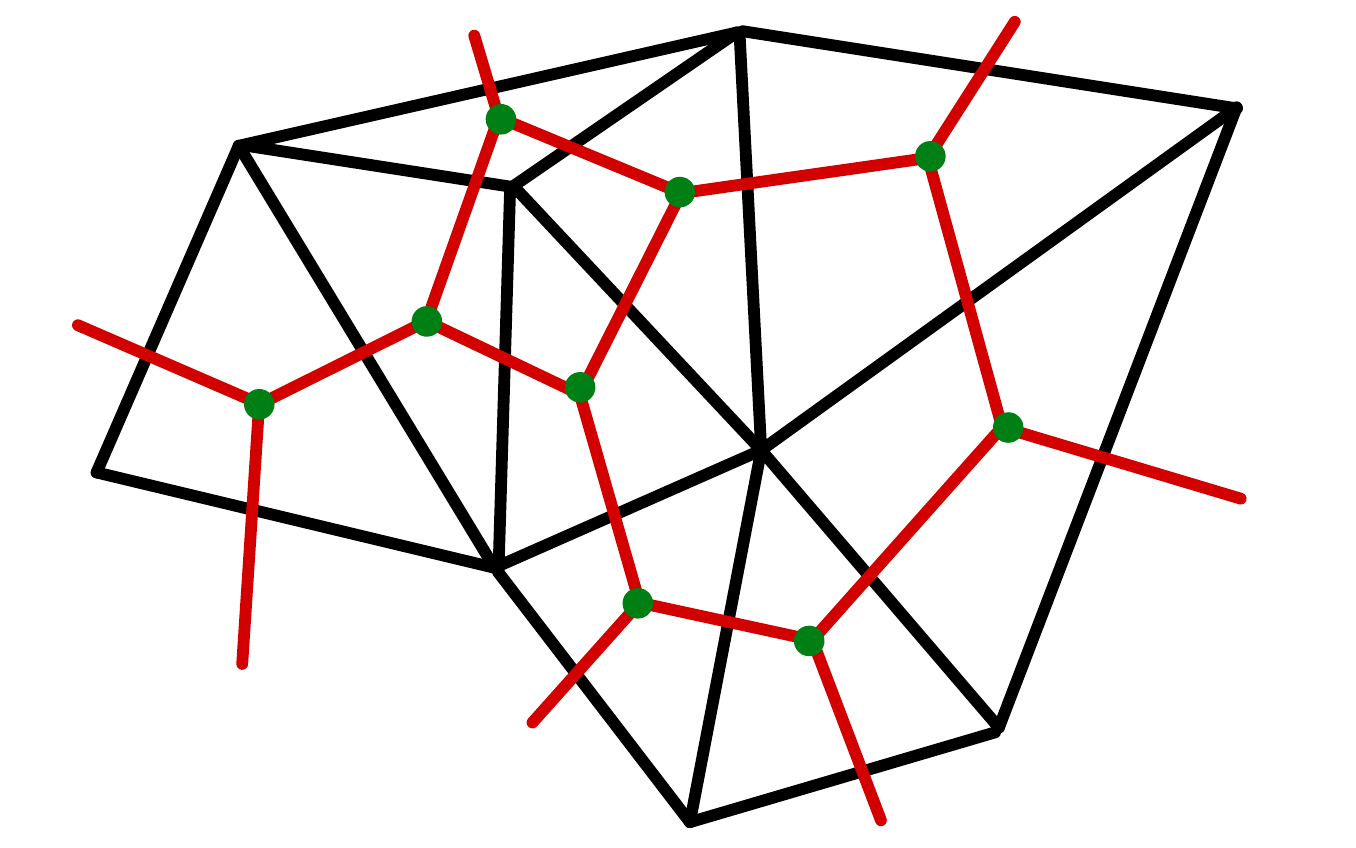}
    \caption{An example of dual skeletons for $d=3$, $q=0$ case. The $(d-q-2)=1$-dimensional skeleton corresponding to the depicted triangulation is shown in black. Its dual, $q=0$-dimensional skeleton is shown in green and corresponds to the dual polyhedral cell decomposition (shown in red). The complement of the $1$-skeleton deformation retracts to the $0$-skeleton and vice versa.}
    \label{fig:Triang-1}
\end{figure}

For the dual skeleton we have the surjective map (the analogue of $g:H^{q+1}(\Sigma,\hat\Sigma)\rightarrow H^{q+1}(\Sigma)$, \textit{not} its Pontryagin dual): 
\begin{equation}
g^*:H^{q^*+1}(\Sigma,\hat{\Sigma}^*;G^*)
\longrightarrow H^{q^*+1}(\Sigma;G^*)=H^{d-q-1}(\Sigma;G^*)\cong H_{q}(\Sigma;G^*)\cong (H^q(\Sigma;G))^*
\label{gstar}
\end{equation}
where the isomorphisms on the right are given by Poincar\'e duality and the general relation between homology and cohomology groups with Pontryagin dual coefficient groups \cite{hurewicz2015dimension}. We also have the isomorphisms
\begin{equation}
    H^{q^*+1}(\Sigma,\hat{\Sigma}^*;G^*)\cong H^{d-q-1}(\Sigma,\sk;G^*)\cong
    H_q(\hat\Sigma;G^*)\cong (H^q(\hat\Sigma;G))^*
\end{equation}
where the first one follows from the homotopy equivalence (\ref{dual-skeletons}) provided by the deformation retraction. The map $g^*$, up to the isomorphisms above, can be understood as the map Pontryagin dual to the previously considered injection $\iota:H^q(\Sigma;G)\rightarrow H^q(\hat\Sigma;G)$.

As was pointed out in Section \ref{sec3}, the vector spaces $\mathcal{H}(\Sigma,b),\;b\in H^{q+1}(\Sigma,\hat{\Sigma})$ of the ungauged theory form a representation $\rho_{\Sigma,b}$ of the abelian group $H^q({\Sigma};G)$.
The Pontryagin dual group $(H^q({\Sigma};G))^*$ classifies its irreducible representations (which are all 1-dimensional). One can then consider the following decomposition of $\mathcal{H}(\Sigma,b)$ into components consisting of different irreducible representations:
\begin{equation}
    \mathcal{H}(\Sigma,b) = \bigoplus_{\tilde{a}\in (H^q({\Sigma};G))^* }\mathcal{H}(\Sigma,b,\tilde{a})
\end{equation}
so that
\begin{equation}
    \left. \rho_b(\beta)\right|_{\Hi(\Sigma,b,\tilde{a})}=e^{2\pi i\tilde{a}(\beta)}\,\mathrm{Id}_{\Hi(\Sigma,b,\tilde{a})}, \qquad \beta\in H^q({\Sigma};G).
\end{equation}
 Consider $a\in H^{q^*+1}(\Sigma,\hat{\Sigma}^*;G^*)$. Using the map and the isomorphisms in (\ref{gstar}) we define the refined Hilbert spaces of the gauged theory as follows:
\begin{equation}
    \mathcal{H}_\text{g}(\Sigma,a)=\bigoplus_{\tilde{b}\in H^{q+1}(\Sigma)}\mathcal{H}(\Sigma,b,g^*(a)).
\end{equation}
For $a=0$ this recovers (\ref{gauged-hilbert-space}), the projection on the trivial representation subspace. Now let $(M,A)$ be a bordism from $(\Sigma_-,a_-)$ to $(\Sigma_+,a_+)$ in the category ${\mathbf{Bord}}_d^{q^*,G^*}$. To be precise, here we consider the version of the category where the objects are  $(d-1)$-manifolds equipped with a skeleton such that there exists its dual in the sense defined above. We have $A\in H^{q^*+1}(M,\hat{\partial M}^*)$ and $a_{\pm}=i_\pm^*(A)\in H^{q^*+1}(\Sigma_\pm,\hat{\Sigma}_\pm^*)$. We then define the "naive" refined version of the value of the gauged TQFT on the bordism as follows:
\begin{equation}
        Z_{\text{g},b_+,b_-}^\tn(M,A) = c(M) \sum_{\substack{B \in \coh[q+1]{M}{\hdm} \\ i^*_{\pm}(B)=b_\pm}} Z(M,B)\,e^{-2\pi i \int_{M}A\cup B}.
\end{equation}
This gives the map from $\Hi(\Sigma_-,b_-)$ to $\Hi(\Sigma_+,b_+)$. As the unrefined functor considered in Section \ref{sec3}, it respects the composition of bordisms, because of the additivity of the integral in the exponential. Note that the cup product $A\cup B$ gives a well-defined class in $H^d(M,\partial M;\R/\Z)$, which one can pair with the fundamental class $[M]\in H_{d}(M,\partial M;\Z)$, because $B\in H^{q+1}(M,\partial M\setminus (\sk_+\sqcup \sk_-);G)$ and $A\in H^{d-q-1}(M,\partial M\setminus  (\sk_+^*\sqcup \sk_-^*);G^*)\cong H^{d-q-1}(M,  \sk_+\sqcup \sk_-;G^*)$, where we used again the homotopy equivalence $\Sigma_{\pm}\setminus \sk_\pm^* \simeq {\sk}_\pm$. For $A=0$ this recovers (\ref{naive-bordism-value}) and for a closed manifold this recovers (\ref{refined-closed-manifold}). As in Section \ref{sec3}, consider the cylinder $\Sigma\times I$, but now decorated by an element $A\in H^{q^*+1}(\Sigma\times I,\hat{\Sigma}_+^*\sqcup ;G^*)\cong H^{q^*+1}(\Sigma,\hat{\Sigma}^*;G^*)\oplus H^{q^*}(\hat{\Sigma}^*;G^*)$ and take it of the form $A=a\oplus 0$. The value of $Z^\text{n}_\text{g}$ on this cylinder then can be identified with the projector on the component $\Hi(\Sigma,b,g^*(a))$ of $\Hi(\Sigma,b)$
This gives the map from $\Hi(\Sigma_-,b_-)$ to $\Hi(\Sigma_+,b_+)$. As the unrefined functor considered in Section \ref{sec3}, it respects the composition of bordisms, because of the additivity of the integral in the exponential. Note that the cup product $A\cup B$ gives a well-defined class in $H^d(M,\partial M;\R/\Z)$, which one can pair with the fundamental class $[M]\in H_{d}(M,\partial M;\Z)$, because $B\in {H^{q+1}(M,\partial M\setminus (\sk_+\sqcup \sk_-);G)}$ and $A\in {H^{d-q-1}(M,\partial M\setminus  (\sk_+^*\sqcup \sk_-^*);G^*)}\cong {H^{d-q-1}(M,  \sk_+\sqcup \sk_-;G^*)}$, where we used again the homotopy equivalence $\Sigma_{\pm}\setminus \sk_\pm^* \simeq {\sk}_\pm$ provided by the deformation retractions. For $A=0$ this recovers (\ref{naive-bordism-value}) and for a closed manifold this recovers (\ref{refined-closed-manifold}). As in Section \ref{sec3}, consider the cylinder $\Sigma\times I$, but now decorated by an element $A\in H^{q^*+1}(\Sigma\times I,\hat{\Sigma}_+^*\sqcup \hat{\Sigma}_-^* ;G^*)\cong H^{q^*+1}(\Sigma,\hat{\Sigma}^*;G^*)\oplus H^{q^*}(\hat{\Sigma}^*;G^*)$ and take it of the form $A=a\oplus 0$. The value of $Z^\text{n}_\text{g}$ on this cylinder then can be identified with the projector on the component $\Hi(\Sigma,b,g^*(a))$ of $\Hi(\Sigma,b)$
\begin{multline}
     Z^\tn_{\tg,b,b}(\Sigma\times I,a\oplus 0)=\frac{1}{|H^{q}(\Si) |} \sum_{\beta \in H^{q}(\Si) } Z(\Si \times I,  b \oplus \iota(\beta)) \,e^{-2\pi i\int (a\oplus 0)\cup (b\oplus \iota(\beta))}=\\
        \frac{1}{|H^{q}(\Si) |}\sum_{\beta \in H^{q}(\Si) }\rho_b(\beta)\,e^{-2\pi i\,g^*(a)(\beta)}  =: P^{g^*(a)}_{\Sigma,b}.
        \label{rep-projection-cylinder}
\end{multline}
By the same argument as in Section \ref{sec3} we can then define the value of the "true" TQFT as
\begin{equation}
    Z_{\tg,b_+,b_-}(M,A):=P_{\Sigma,b_+}^{g^*(a_+)}\left.Z^\tn_{\tg,b_+,b_-}(M,A)\right.|_{\Hi(\Sigma,b_-,g^*(a_-))},\qquad a_\pm=i_\pm^*(A)
\end{equation}
which gives the map from the component $\Hi(\Sigma,b_-,g^*(a_-))$ of $\Hi_\tg(\Sigma,a_-)$ to the component $\Hi(\Sigma,b_+,g^*(a_+))$ of $\Hi_\tg(\Sigma,a_+)$.

As in the unrefined case, the definition of $\Hi(\Sigma,a)$ relies on the choice choice of the preimages $b\in H^{q+1}(\Sigma,\hat{\Sigma})$ for each $\tilde{b}\in H^{q+1}(\Sigma)$ (now the skeleton $\sk$ is fixed by $\sk^*$ up to the choice of the deformation retraction, so the group is $H^{q+1}(\Sigma,\hat{\Sigma})$ fixed). One can again argue that different choices lead to naturally isomorphic functors. Similarly to the unrefined case, the natural isomorphism is provided by the isomorphisms
\begin{equation}
Z(\Sigma\times I,B)\,e^{-2\pi i\int_{\Sigma\times I}(a\oplus 0)\cup B}\;:\;\Hi(\Sigma,b,g^*(a))\stackrel{\cong}{\longrightarrow} \Hi(\Sigma,b',g^*(a)).
\end{equation}
They are independent of the choice of $B$ and thus canonical.

The vector space $\Hi_\tg(\Sigma,a)$ form a representation $\rho_{\Sigma,a}^*$ of $H^{q^*}(\Sigma)$, as previously, provided by 
\begin{equation}
    \rho^*_{\Sigma,a}(\alpha)=Z_\tg(\Sigma\times I,a\oplus \iota^*(\alpha))
\end{equation}
where $\iota^*:H^{q^*}(\Sigma)\rightarrow H^{q^*}(\hat\Sigma^*)$ is the analogue of the map $\iota$. The generalization (\ref{rep-projection-cylinder}) to the case of non-trivial $\alpha$ gives us:
\begin{multline}
    \rho_b(\alpha)|_{\Hi(\Sigma,b,g^*(a))}=Z_{\tg,b,b}(\Sigma\times I,a\oplus \iota^*(\alpha))=\\
    =e^{-2\pi i\int \iota^*(\alpha)\cup b }\,\mathrm{Id}_{\Hi(\Sigma,b,g^*(a))}=e^{-2\pi i g(b)(\alpha) }\,\mathrm{Id}_{\Hi(\Sigma,b,g^*(a))}
\end{multline}
where we identify $g:H^{q+1}(\Sigma,\hat{\Sigma};G)\rightarrow H^{q+1}(\Sigma;G)$ with the map Pontryagin dual to $\iota^*$ and use the isomorphism $(H^{q^*}(\Sigma;G^*))^*\cong H^{q+1}(\Sigma;G)$. It follows that the components $\Hi(\Sigma,b,g^*(a))$ consists of irreducible representation corresponding to $g(b)$. From this one can explicitly see that repeating the gauging once again (with a change of some sign conventions) recovers the original functor up to a natural isomorphism. For the objects we have:
\begin{equation}
    \Hi_{\tg \tg}(\Sigma,b')=
    \bigoplus_{\substack{\tilde{a}\in H^{q^*}(\Sigma,G)\\ (g^*(a)=\tilde{a})}}\Hi_\tg(\Sigma,a)|_{\text{irreps }g(b')}=
     \bigoplus_{\tilde{a}\in H^{q^*}(\Sigma,G)}\Hi(\Sigma,b,\tilde{a})=\Hi(\Sigma,b)
\end{equation}
where $b$ is the selected preimage of $\tilde{b}=g(b')$ (note that by construction $\sk^{**}\simeq \sk$ so that $b$ and $b'$ can be considered as elements of the same group $H^{q+1}(\Sigma,\hat{\Sigma})$), in the first gauging procedure. Therefore $g(b')=g(b)$ and indeed $\Hi_{\tg\tg}(\Sigma,b')\cong \Hi(\Sigma,b')$. The value on morphisms can be recovered (up to a value of the invertible TQFT with the value $\prod_{i=0}^d|H^{i}(M;G)|^{(-1)^i}$ on a closed manifold) using the fact that
\begin{equation}
    \sum_{A\in H^{q^*+1}(M,\hat{\partial M}^*;G)}e^{2\pi i \int (B-B')\cup A}\propto \delta_{B,B'},\qquad B,B'\in H^{q+1}(M,\hat{\partial M}).
\end{equation}

\section*{Acknowledgements}
PP would like to thank Francesco Benini, Christian Copetti, Francesco Costantino for useful discussions on related subjects. MF would like to thank Nils Carqueville for insightful discussions on the topic.

\begin{appendices}

\section{Details on (co)homology groups}\label{secA1}

In this appendix we collect various properties of (co)homology groups describing the gauge fields. We use the standard results that can be found in \cite{hat,maunder1996algebraic}.

\subsection{Fields on the boundaries \texorpdfstring{$\Si$}{}}
We shall start by analyzing the long exact sequence of cohomology groups describing fields on a boundary $(d-1)$-manifold $\Si$. We have the following long exact sequence of cohomology groups associated to the pair $(\hat\Sigma,\Sigma)$:
\begin{center}
    \begin{tikzcd}
        0 \arrow[r] & \coh[0]{\Si}{\hSi} \arrow[r] & H^0(\Si) \arrow[r] & H^0(\hSi) \arrow[r] & {} \\
        {} \arrow[r] & \cdots & {} & \cdots \arrow[r] & {} \\
        {} \arrow[r] & \coh[q]{\Si}{\hSi} \arrow[r] & H^q(\Si) \arrow[r,"\iota"] & H^q(\hSi) \arrow[r] & {} \\
        {} \arrow[r] & \coh[q+1]{\Si}{\hSi} \arrow[r,"g"] & H^{q+1}(\Si) \arrow[r] & H^{q+1}(\hSi) \arrow[r] & \cdots 
    \end{tikzcd}
\end{center}
Here, thanks to Poincar\'e-Alexander-Lefschetz duality and isomorphisms between singular and cellular homology we can see:
\begin{align*}
    \coh[k]{\Sigma}{\hSi} &= H_{d-1-k}({\sk}) = 0 &\text{ for any } k \leq q \\
    H^{k}(\hSi) &= H_{d-1-k}(\Si,{\sk}) = 0 &\text{ for any } k \geq q+1 \\
    H^k(\Sigma) &= H_{d-k-1}(\Sigma)= H_{d-k-1}(\Sigma,\sk)=H_{d-k-1}(\Sigma)=H^k(\hat{\Sigma}) &\text{ for any } k < q. \\
\end{align*}
The exact sequence tells us that in general $H^{q+1}(\Si) \neq \coh[q+1]{\Si}{\hSi}$, but there is only a surjective map $g:H^{q+1}(\Si) \rightarrow \coh[q+1]{\Si}{\hSi}$. There is also an injective map $\iota:H^q(\Sigma)\rightarrow H^{q}(\hat\Sigma)$. In particular, we can separate the following exact sequence:
\begin{equation}
0\longrightarrow H^q(\Sigma)\stackrel{\iota}{\longrightarrow} H^q(\hat\Sigma) \longrightarrow H^{q+1}(\Sigma,\hat\Sigma)
\stackrel{g}{\longrightarrow} H^{q+1}(\Sigma) \longrightarrow 0.
\end{equation}
This will be important for analyzing the dependence of the physical Hilbert space on the background gauge fields on a spatial slice $\Sigma$.

\subsection{Relating fields on \texorpdfstring{$M$}{the bordism} to fields on \texorpdfstring{$M_+$ and $M_-$}{on the boundaries}}
For understanding the \textit{composition} of bordisms $(M_-,B_-):(\Sigma_-,b_-)\rightarrow (\sigma,b_\sigma)$ and $(M_+,B_+):(\sigma,b_\sigma)\rightarrow (\Sigma,b_+)$ into $(M,B):(\Sigma_-,b_-)\rightarrow (\Sigma,b_+)$ we look at the following relative Mayer-Vietoris sequence:
\begin{center}
    \begin{tikzcd}
    0 \arrow[r] & \coh[0]{M}{\hSi} \arrow[r] & \coh[0]{M_+}{\hsi_+} \bigoplus \coh[0]{M_}{\hsi_-} \arrow[r] & 0 \arrow[r] & {}\\
    {} \arrow[r] & \cdots & {} & \cdots \arrow[r] & {}\\
    {} \arrow[r] & \coh[q]{M}{\hSi} \arrow[r] & \coh[q]{M_+}{\hsi_+} \bigoplus \coh[q]{M_}{\hsi_-} \arrow[r] & 0 \arrow[r] & {}\\
    {} \arrow[r] & \coh[q+1]{M}{\hSi} \arrow[r,"\phi"] & \coh[q+1]{M_+}{\hsi_+} \bigoplus \coh[q+1]{M_}{\hsi_-} \arrow[r,"\psi"] & H^{q+1}(\sigma,\hsi) \arrow[r] & \cdots
\end{tikzcd}
\end{center}
Here $\hsi_\pm := \hsi \sqcup \Si_\pm$ and $\hat\Sigma:=\hat\Sigma_+\sqcup \hat\sigma\sqcup \hat\Sigma_-$ and we have used that $H^{k}(\sigma,\hat{\sigma})\cong 0$ for $k\leq q$ as was argued above. From this sequence, we observe that there is an isomorphism $\coh[k]{M}{\hSi} \simeq \coh[k]{M_+}{\hsi_+} \oplus \coh[k]{M_-}{\hsi_-}$ for $k = 0, \ldots, q$. However, for the case of $\coh[q+1]{M}{\hSi}$, the isomorphism does not hold, and we need to understand this case more carefully.
We notice that the map $\psi$ sends $(B_+, B_-) \mapsto i_-^*(B_-) - i_+^*(B_+)$, where $i_-^*$ and $i_+^*$ are the pullbacks induced by the embeddings of the cut $\sigma$ into $M_\pm\subset M$. We require that $B_+$ and $B_-$ agree on the cut, meaning that $i_-^*(B_-) - i_+^*(B_+) = 0$. This implies that $(B_+, B_-) \in \text{Ker}\, \psi = \text{Im}\, \phi$. Since $\phi$ is injective, there exists a unique $\overline{B} \in \coh[q+1]{M}{\hSi}$ such that $\phi(\overline{B}) = (B_-, B_+)$ satisfies that condition. This means that even though we do not have an isomorphism, the condition to agree on the cut eliminates any ambiguity in relating $\overline{B}$ to $B_\pm$.

We, however, want to define the field $B \in \coh[q+1]{M}{\hdm}$ in terms of $B_\pm$. To address this, we consider the long exact sequence for the triple $(\hdm, \hSi, M)$ where we can relate $\coh[q+1]{M}{\hSi}$ and $\coh[q+1]{M}{\hdm}$. The sequence is given by:
\begin{center}
    \begin{tikzcd}
    0 \arrow[r] & \coh[0]{M}{\hSi} \arrow[r] & \coh[0]{M}{\hdm} \arrow[r] & H^0(\hsi) \arrow[r] & {}\\
    {} \arrow[r] & \cdots & {} & \cdots \arrow[r] & {}\\
    {} \arrow[r] & \coh[q]{M}{\hSi} \arrow[r] & \coh[q]{M}{\hdm} \arrow[r] & H^q(\hsi) \arrow[r,"\gamma"] & {}\\
    {} \arrow[r,"\gamma"] & \coh[q+1]{M}{\hSi} \arrow[r,"\eta"] & \coh[q+1]{M}{\hdm} \arrow[r] & 0 \arrow[r] & \cdots
\end{tikzcd}
\end{center}
From this sequence, we observe that the map $\eta$ is surjective, and we have:
\begin{equation}
    \coh[q+1]{M}{\hdm} = \bigslant{\coh[q+1]{M}{\hSi}}{\text{Ker}\eta}
\end{equation}
where
\begin{equation}
    |\text{Ker}\eta|=|\text{Im}\gamma| = \prod_{i=0}^{q} \left( \frac{|H^i(\hsi)| |\coh[i]{M}{\hSi}|}{|\coh[i]{M}{\hdm}|} \right)^{(-1)^{q-i}}.
\end{equation}

\subsection{Splitting fields on the cylinder into parallel and perpendicular components}
We now proceed to study the long exact sequence of the triple
\begin{equation}
       2\sk\equiv {\sk}_+\sqcup {\sk}_-  \subset {\sk} \times I \subset \Si \times I
\end{equation}
where $\sk_\pm$ is two copies of the same skeleton $\sk\subset\Sigma$ in $\Sigma_-=\Sigma\times\{0\}$ and $\Sigma_+=\sigma\times\{1\}$. We start by examining the short exact sequence of chain groups        
    \begin{equation}
        0 \longrightarrow C_n({\sk} \times I, 2{\sk}) \longrightarrow C_n(\Si \times I, 2{\sk}) \longrightarrow C_n(\Si \times I, {\sk} \times I) \longrightarrow 0
    \end{equation}
     which gives rise to the long exact sequence of homology groups:
    \begin{equation}
        \cdots \longrightarrow H_n({\sk} \times I, 2{\sk}) \longrightarrow H_n(\Si \times I, 2{\sk}) \longrightarrow H_n(\Si \times I, {\sk} \times I) \longrightarrow \cdots
    \end{equation}
Using the isomorphism between the singular and cellular homology groups, we can consider $C_n$ to be cellular chain complexes. On $\Sigma$ we use a cellular decomposition that defines the $(d-q-2)$-skeleton $\sk\subset \Sigma$. We can then consider the product cellular structure on $\Sigma\times I$ where $k$-cells consist of: (i)  $I=[0,1]$ times $(k-1)$-cells in $\Sigma$, (ii) $\{0\}$ or $\{1\}$ times $k$-cells in $\Sigma$.
    
The groups of interest in our case are:       
    \begin{align*}
        H_i({\sk} \times I, 2{\sk}) & = 0 &\text{ if } i > d-q-1 \\
        \homo[i]{\Si \times I}{2{\sk}} & &\\
         H_i(\Si \times I, {\sk} \times I) &\cong \homo[i]{\Si}{{\sk}} \\
         &=H_i(\Si) &\text{ if } i > d-q-1
    \end{align*}
where we used homotopy equivalence $(\Sigma\times I,\sk\times I)\simeq (\Sigma,\sk)$ that can be provided by the deformation retraction of $\Sigma\times I$ onto $\Sigma_+$ or $\Sigma_-$.

    Although our main interest lies in the $(q+1)$-th cohomology group (or $(d-q-1)$-th homology group), we will analyze the full sequence for later use. Using the isomorphisms above we have:
    \begin{center}
            \begin{tikzcd}
            0 \arrow[r] & 0 \arrow[r]
            & \homo[d]{\Si \times I}{2{\sk}} \arrow[r, "\cong"] & 0 \arrow[r] & {}\\
            {} \arrow[r] & 0 \arrow[r] & \homo[d-1]{\Si \times I}{2{\sk}} \arrow[r, "\cong"] & H_{d-1}(\Si) \arrow[r] & {} \\
            {} \arrow[r] & \cdots & {} & \cdots \arrow[r] & {} \\ 
            {} \arrow[r] & 0 \arrow[r] & \homo[d-q]{\Si \times I}{2\sk} \arrow[r,] & H_{d-q}(\Si) \arrow[r, "f"] & {} \\
            {} \arrow[r,"f"] &  \homo[d-q-1]{{\sk} \times I}{2\sk} \arrow[r] & \homo[d-q-1]{\Si \times I}{2\sk} \arrow[r] & H_{d-q-1}(\Si,{\sk}) \arrow[r,"h"] & {} \\
            {} \arrow[r,"h"] &  \homo[d-q-2]{{\sk} \times I}{2\sk} \arrow[r] & \homo[d-q-2]{\Si \times I}{2\sk} \arrow[r] & H_{d-q-2}(\Si,{\sk}) \arrow[r] & \cdots 
        \end{tikzcd}
    \end{center}
    
   Now we analyze the maps $f$ and $g$ explicitly in terms of the short exact sequence of chain complexes and show that they are zero maps. Its relevant part has the following form:
           \begin{tikzcd}
           0 \arrow[r]
           & 0 \arrow[d] \arrow[r, "i_{d-q}"]
           & C_{d-q}(\Si \times I, 2\sk) \arrow[r, "j_{d-q}"] \arrow[d, "\partial"]
           & C_{d-q}(\Si \times I,{\sk} \times I) \arrow[r] \arrow[d]
           & 0\\
           0 \arrow[r]
           & C_{d-q-1}({\sk} \times I, 2\sk) \arrow[d] \arrow[r, "i_{d-q-1}"]
           & C_{d-q-1}(\Si \times I, 2\sk) \arrow[r, "j_{d-q-1}"] \arrow[d, "\partial"]
           & C_{d-q-1}(\Si \times I,{\sk} \times I) \arrow[r] \arrow[d]
           & 0\\
           0 \arrow[r]
           & C_{d-q-2}({\sk} \times I, 2\sk) \arrow[r, "i_{d-q-2}"]
          & C_{d-q-2}(\Si \times I, 2\sk) \arrow[r, "j_{d-q-2}"]
          & C_{d-q-2}(\Si \times I,{\sk} \times I) \arrow[r]
          & 0
      \end{tikzcd}
 For the map $f$ we have $f([\delta])=[\omega]$ where $j_{d-q}(\varepsilon)=\delta$ and $\partial \varepsilon = i_{d-q-1}(\omega)$. Here $\delta$ is a representative of a class in $\homo[d-q]{\Si \times I}{{\sk} \times I} \cong \homo[d-q]{\Sigma}{{\sk}} = H_{d-q}(\Sigma)$. From the isomorphisms it follows that one can take $\delta$ to be a $(d-q)$-cycle in $\Si\times \{0\}$. We can then simply take $\varepsilon=\delta\in C_{d-q}(\Si \times I, 2\sk)$. Then $\partial \varepsilon =0$, which implies $\omega = 0$. Therefore $f=0$.
 
    A similar reasoning works for the map $h$. We have $h([\delta])=[\omega]$  where $j_{d-q-1}(\varepsilon)=\delta$ and $\partial \varepsilon = i_{d-q-2}(\omega)$.   Here $\delta$ is a representative of a class in $\homo[d-q-1]{\Si \times I}{{\sk} \times I} \cong \homo[d-q-1]{\Sigma}{{\sk}}$. Therefore one can take $\delta$ to be  $(d-q-1)$-cycle in $\Si\times\{0\}$ relative to ${\sk}\times\{0\}$. We can then simply take $\varepsilon=\delta\in C_{d-q-1}(\Si \times I, 2\sk)$ . Then $\partial \varepsilon =0$, which implies $\omega = 0$. Therefore  $h=0$.\bigskip

    This leads us to the following short exact sequence:
        \begin{center}
        \begin{tikzcd}
            0 \arrow[r]
            & \homo[d-q-1]{{\sk} \times I}{2\sk} \arrow[r]
            & \homo[d-q-1]{\Si \times I}{2\sk} \arrow[r] 
            & H_{d-q-1}(\Si,\sk) \arrow[r] 
            & 0\\
        \end{tikzcd}
    \end{center}
The sequence is split by the map $H_{d-q-1}(\Si,\sk)\rightarrow H_{d-q-1}(\Sigma\times I,2\sk)$ induced by the inclusion $(\Sigma,\sk)\hookrightarrow (\Sigma\times I,2\sk)$ which identifies $\Sigma$ with $\Sigma\times \{0\}$. 
    
    As a consequence, using Poincar\'e-Alexander-Lefschetz daulity we can now establish the following relationships:    
    \begin{align*}
        \coh[0]{\Si \times I}{\hSi_+\sqcup \hSi_-} &\cong \homo[d]{\Si \times I}{2\sk} = 0 \\
        \coh[k]{\Si \times I}{\hSi_+\sqcup \hSi_-} &\cong \homo[d-k]{\Si \times I}{2\sk} \cong H_{d-k}(\Si) \cong H^{k-1}(\Si) \text{ if } k=1,..,q \\
        \coh[q+1]{\Si \times I}{\hSi_+\sqcup \hSi_-} &\cong \homo[d-q-1]{\Si \times I}{2\sk} \cong \homo[d-q-1]{{\sk} \times I}{2\sk} \oplus H_{d-q-1}(\Si)  \\
        &\cong H^{q+1}(\Si,\hSi) \oplus H^{q}(\hSi)
    \end{align*}
    where we used that $\homo[d-q-1]{{\sk} \times I}{\sk_\pm} \cong H_{d-q-2}({\sk}) \cong \coh[q+1]{\Si}{\hSi} $. This isomorphism can be argued by considering the long exact sequence for the pair $(2\sk,\sk\times I)$:
    \begin{equation}
\begin{tikzcd}
 \ldots \ar[r] &  H_{d-q-1}(\sk)\oplus H_{d-q-1} (\sk) \ar[r] &  H_{d-q-1}(\sk\times  I) \ar[r] &   H_{d-q-1}(\sk\times I,2\sk) \ar[r] &  {}  \\
& 0 \ar[u,equal] &   0 \ar[u,equal] &    &  
\\
 \ar[r] &  H_{d-q-2}(\sk)\oplus H_{d-q-2} (\sk)\ar[r,"r"] & H_{d-q-2}(\sk\times I) \ar[r] &   H_{d-q-2}(\sk\times I,2\sk)  \ar[r] & \ldots  \\
& &   H_{d-q-2}(\sk) \ar[u,equal] &      &  
\end{tikzcd}
\end{equation}
The map $r$ can be understood as the sum map and therefore $H_{d-q-1}(\sk\times I,2\sk)\cong \mathrm{Ker}\, r\cong H_{d-q-2}(\sk)$.

\end{appendices}

\bibliography{mybibliography}
\bibliographystyle{JHEP}

\end{document}